\documentclass[a4paper, 10pt]{article}
\usepackage[margin=.95in,dvips]{geometry}
\usepackage[bottom]{footmisc}
\usepackage[utf8]{inputenc}
\usepackage{lmodern}
\usepackage{indentfirst}
\usepackage[affil-it]{authblk}
\usepackage{tocloft}
\usepackage{titlesec}
 \titleformat{\subparagraph}[runin]{\normalfont}{\thesubparagraph}{0pt}{\underline}[.]
 \titleformat{\paragraph}[hang]{\normalfont\bfseries}{\theparagraph}{0pt}{}
\usepackage{titletoc}

\usepackage[textsize=scriptsize]{todonotes} 
\makeatletter
\renewcommand{\todo}[2][]{\tikzexternaldisable\@todo[#1]{#2}\tikzexternalenable}
\makeatother

\usepackage{standalone}
\usepackage{subfiles}

\usepackage[hidelinks,bookmarksopen,bookmarksdepth=3]{hyperref}
\usepackage{enumitem}
\setlist[enumerate]{itemsep=2.0pt plus 1.0 pt minus 0.5pt, topsep=4.0pt plus 2.0 pt minus 1.0pt}
\setlist[itemize]{itemsep=2.0pt plus 1.0 pt minus 0.5pt, topsep=4.0pt plus 2.0 pt minus 1.0pt}
\usepackage{tikz}
\usetikzlibrary{external}
\usepackage{calc}
\usepackage{graphicx}
\usepackage{subcaption}
\usepackage{epsfig}
\usepackage{color, soul}

\usepackage{marginnote}

\usepackage{color,soulutf8}

\usepackage{amsmath}
\usepackage{amsthm, slashed}
\usepackage[capitalize]{cleveref}
\usepackage{amssymb}
\usepackage{array}
\usepackage{xfrac}
\usepackage{bbm}
\usepackage{bm}
\usepackage{marvosym}
\usepackage{revsymb} 
\usepackage{mathtools}
\usepackage{mathrsfs}
\usepackage{upgreek}
\usepackage{accents}
\usepackage{textcomp}
\usepackage{slashed} 
\usepackage{harmony} 
\usepackage{wasysym} 
\wasyfamily
\DeclareFontShape{U}{wasy}{b}{n}{ <-10> ssub * wasy/m/n
 <10> <10.95> <12> <14.4> <17.28> <20.74> <24.88>wasyb10 }{}
\DeclareMathAlphabet\mathbfcal{OMS}{cmsy}{b}{n}
\newcommand\numberthis{\addtocounter{equation}{1}\tag{\theequation}}

\renewcommand{\Re}{\operatorname{Re}}
\renewcommand{\Im}{\operatorname{Im}}
\DeclareMathOperator{\sign}{sign}
\allowdisplaybreaks

\numberwithin{equation}{section}

\newtheorem{bigtheorem}{Theorem}

\newtheorem{theorem}{Theorem}[section]
\newtheorem*{theorem*}{Theorem}

\newtheorem*{conjecture*}{Conjecture}
\newtheorem{corollary}[theorem]{Corollary}
\newtheorem*{corollary*}{Corollary}

\newtheorem{definition}[theorem]{Definition}
\newtheorem{proposition}[theorem]{Proposition}
\newtheorem{lemma}[theorem]{Lemma}
\newtheorem{remark}[theorem]{Remark}

\newtheoremstyle{mystyle}
  {}
  {}
  {\itshape}
  {}
  {\bfseries}
  {.}
  { }
  {\thmname{#1}\thmnumber{ #2}\thmnote{ (#3)}}

\theoremstyle{mystyle}

\usepackage{filecontents}

\newcommand{\mc}[1]{\mathcal{#1}}
\newcommand{\mr}[1]{\mathrm{#1}}

\newcommand{\p}{\partial}
\newcommand{\lp}{\left}
\newcommand{\rp}{\right}

\newcommand{\swei}[2]{#1^{[{#2}]}}

\newcommand{\smlambda}[2]{#1_{m,\lambdabar}^{[{#2}],\,a,\omega}}
\newcommand{\smlambdaXi}[2]{#1_{m,\lambdabar}^{[{#2}],\,\Xi,a,\omega}}

\usepackage{comment}

\title{\textbf{Hidden spectral symmetries and mode stability}\\ \textbf{ of subextremal Kerr(-de Sitter) black holes}}

\author[1,2]{{\Large Marc Casals}}
\author[3,4]{{\Large  Rita \mbox{Teixeira da Costa}\vspace{0.4cm}}}

\affil[1]{\small  Centro Brasileiro de Pesquisas Físicas (CBPF),
Rio de Janeiro,
CEP 22290-180, Brazil  }
\affil[2]{School of Mathematics and Statistics, University College Dublin, Belfield, Dublin 4, Ireland \protect \\
{\small\tt{mcasals@cbpf.br}, \tt{marc.casals@ucd.ie}}\vspace{0.2cm}\ }

\affil[3]{\small 
Princeton Gravity Initiative \& Department of Mathematics, Princeton University, NJ 08540, United States}
\affil[4]{\small 
University of Cambridge, Center for Mathematical Sciences, Cambridge CB3 0WA, United Kingdom\protect \\
{\small\tt{rita.t.costa@princeton.edu}, \tt{rita.t.costa@dpmms.cam.ac.uk}}}

\date{\today}

\begin{document}
\maketitle

\vspace{-0.5cm}

\begin{abstract}
We uncover hidden spectral symmetries of the Teukolsky equation in Kerr(-de Sitter) black holes, recently conjectured by Aminov, Grassi and Hatsuda (Ann. Henri Poincaré, and Gen.\ Relativ.\ Grav., 53(10):93, 2021) in the zero cosmological constant case. Using these symmetries, we provide a new, simpler proof of mode stability for subextremal Kerr black holes. We also present a partial mode stability result for Kerr-de Sitter black holes.
\end{abstract}

\section{Introduction}

In General Relativity, a vacuum spacetime is a $(1+3)$-dimensional Lorentzian manifold solving the Einstein equations
\begin{equation}
\mr{Ric}(g)=\Lambda g\,, \label{eq:Einstein-cosmological-constant}
\end{equation}
where $g$ is the Lorentzian metric and $\Lambda$ is the cosmological constant. Here, we will focus especially on the $\Lambda\geq 0$ case. Of paramount importance are the Kerr and Kerr-de Sitter,  black hole families of solutions to \eqref{eq:Einstein-cosmological-constant} with $\Lambda=0$ and $\Lambda>0$, respectively. These are parametrized by their mass $M>0$ and a specific angular momentum $a\in\mathbb{R}$ which is constrained in terms of $M$ and $\Lambda$; for instance, in the case $\Lambda=0$, Kerr black holes verify the bound $|a|\leq M$.

As Kerr(-de Sitter) black holes are stationary spacetimes, they correspond to \textit{equilibrium states} for \eqref{eq:Einstein-cosmological-constant}, and one would like to determine whether they are \textit{stable} or \textit{unstable} equilibria \cite{Regge1957}. A great deal of progress on the problem has been achieved for the spherically symmetric non-rotating ($a=0$) subfamily and perturbations thereof. In the $\Lambda>0$ case, Hintz and Vasy showed a full nonlinear stability statement when the black hole parameters are such that $|a|\ll M,\Lambda$. In the more nuanced $\Lambda=0$ setting, a complete picture of the nonlinear stability of the $a=0$ subfamily was only very recently established by Dafermos, Holzegel, Rodnianski and Taylor in \cite{Dafermos2021}; see also \cite{Klainerman2021} for some progress in the direction of an extension to $|a|\ll M$ black holes. 

Outside these special classes, stability of Kerr(-de Sitter) black holes has remained an open question. Nevertheless, the previous works lay out a clear roadmap for investigating it. The key step in the program is to understand the so-called Teukolsky equation with $s=\pm 2$, as it describes the dynamics of some \textit{gauge-invariant} curvature components in the linearized Einstein equations around Kerr(-de Sitter) black holes. The Teukolsky equation was first obtained in the $\Lambda=0$ case by Teukolsky \cite{Teukolsky1973}, and later for $\Lambda>0$ in \cite{Khanal1983}. For $\Lambda\geq 0$, writing $\Xi=1+a^2\Lambda/3$,  this equation takes the form
\begin{equation}
\begin{split}
\Box_g\upalpha^{[s]} &+\frac{s}{\rho^2\Xi^2}\frac{d\Delta}{dr}\p_r\upalpha^{[s]} +\frac{2s}{\rho^2\Xi}\lp[\Xi\frac{a}{2\Delta}\frac{d\Delta}{dr}+i\frac{\cos\theta}{\sin^2\theta}-\frac{ia^2\Lambda}{3\Delta_\theta}\cos\theta\rp]\p_\phi\upalpha^{[s]} 
\\
&+\frac{2s}{\rho^2}\lp[\frac{(r^2+a^2)}{2\Delta}\frac{d\Delta}{dr}-2r -ia\frac{\cos\theta}{\Delta_\theta}\rp]\p_t\upalpha^{[s]} \\
&+\frac{s}{\rho^2\Xi^2}\lp[1-\frac{a^2\Lambda}{3}-\frac{6(3+2s)r^2}{\Lambda}-\Xi^2 s\frac{\cot^2\theta}{\Delta_\theta}\rp]\upalpha^{[s]} =\frac{2\Lambda}{3}\upalpha^{[s]}\,, 
\end{split}\label{eq:Teukolsky-equation-intro}
\end{equation}
in Boyer-Lindquist-type coordinates, where  $\Delta$ is a function of $r$ given by 
\begin{equation}
\Delta=(r^2+a^2)\lp(1-\frac{\Lambda r^2}{3}\rp)-2Mr\,; \label{eq:Delta-intro}
\end{equation}
see already Sections~\ref{sec:geometry} and \ref{sec:geometry-KdS} for the definitions of $\Delta_\theta$ and $\rho$. In \eqref{eq:Teukolsky-equation-intro}, $\Box_g$ is the covariant wave operator on the fixed Kerr(-de Sitter) metric. The parameter $s$ is called a spin-weight taking values in $\frac12\mathbb{Z}$. Aside from the important case $s=\pm 2$ concerning gravitational perturbations, \eqref{eq:Teukolsky-equation-intro} gives the dynamics of some gauge-invariant electromagnetic components in the linearized Maxwell equations if $s=\pm 1$, and describes perturbations by Dirac fields if $s=\pm \frac12$ and by conformal scalar fields if $s=0$. Note that a conformal scalar field is massless in the $\Lambda=0$ setting, but has a specific Klein--Gordon mass in the $\Lambda\neq 0$ case, see Section~\ref{sec:radial-ODE-KdS} for further insight regarding our choice of mass.

If Kerr-(de Sitter) black holes are \textit{nonlinearly stable}, the most basic statement we can hope to prove for \eqref{eq:Teukolsky-equation-intro} is that it is \textit{modally stable}, i.e.\ that there are no separable solutions to \eqref{eq:Teukolsky-equation-intro} which are exponentially growing or bounded but non-decaying in time. By separable solution we mean a solution of the form
\begin{align} \label{eq:separable-ansatz-intro}
\swei{\upalpha}{s}(t,r,\theta,\phi)=e^{-i\omega t}e^{im\phi}{S}(\theta) {\upalpha}(r)\,,
\end{align}
where $S$ and $\upalpha$ satisfy, respectively, an angular ODE and a radial ODE with suitable boundary conditions; note that $\swei{\upalpha}{s}$ is  exponentially growing if $\Im\omega>0$ and bounded, non-decaying in time if $\omega\in\mathbb{R}$.  Our motivation for studying solutions as in \eqref{eq:separable-ansatz-intro} comes from  Carter's result \cite{Carter1968} that $\Box_g$, and hence \eqref{eq:Teukolsky-equation-intro}, is separable.

The only systematic way of establishing mode stability for a PDE is to find a conserved \textit{coercive} energy. In a stationary spacetime such as Kerr(-de Sitter) we have an obvious candidate for such an energy: the conserved quantity associated to the stationary Killing field. This energy is coercive for non-rotating, i.e.\ $a=0$, black holes, thus mode stability follows at once in such spacetimes. Perturbative arguments, relying on the celebrated redshift effect of Dafermos and Rodnianski \cite{Dafermos2009}, allow us to extend the $a=0$ mode stability result to the larger class of black holes which are very slowly rotating, i.e.\ where $|a|\ll M$ \cite{Dafermos2010} and, if $\Lambda>0$, $|a|\ll M,\Lambda$ \cite{Dyatlov2011a}. 

In the general $a\neq 0$ case, this approach breaks down completely. The conserved energy associated to the Killing field is generally \textit{non-coercive}. In the cases $s=0, \pm 1,\pm 2$ the non-coercivity goes by the name of {superradiance}. From the point of view of the black hole geometry, superradiance is a consequence of the fact that rotating black holes have {ergoregions} where the stationary field becomes spacelike. Under the separable ansatz \eqref{eq:separable-ansatz-intro}, superradiance is captured by a simple condition on the frequency parameters $\omega$ and $m$. For instance, take $\omega\in\mathbb{R}$: if $\Lambda=0$, the superradiant condition reads
\begin{equation}
m\neq 0\,, \quad 0<\frac{\omega}{m}<\frac{a}{r_+^2+a^2} \,,\label{eq:superradiance-intro}
\end{equation}
where $r_+$ is the largest root of \eqref{eq:Delta-intro}; if $\Lambda>0$, it reads
\begin{align}
m\neq 0\,, \quad \frac{a}{r_2^2+a^2}< \frac{\omega}{m}< \frac{a}{r_1^2+a^2}\,,\label{eq:superradiance-KdS-intro}
\end{align}
where $r_2,r_1$ are, respectively, the largest and second largest roots of \eqref{eq:Delta-intro}. Furthermore, the stabilizing effect of redshift is generally not strong enough to overcome superradiance. 

Superradiance, therefore, emerges as an important obstacle in establishing mode stability and, indeed, black hole stability for general back hole parameters: there are several examples where superradiance leads to mode \textit{instability}. Massive scalar fields on Kerr can produce a black hole bomb \cite{Shlapentokh-Rothman2015a}, and even milder modifications of the scalar potential can lead to non-decaying modes \cite{Moschidis2017b} in Kerr. For Kerr's $\Lambda<0$ cousin, Kerr-Anti de Sitter, superradiance may be even more damaging to its stability\footnote{It is important to note that there may be other factors at play when it comes to stability problems for solutions to \eqref{eq:Einstein-cosmological-constant} with $\Lambda<0$. Indeed, even the trivial solution, Anti-de Sitter space, has been shown to be unstable under additional coupling to several matter models in spherical symmetry \cite{Moschidis2017a,Moschidis2018}.}: massive and massless scalar fields also admit exponentially growing modes \cite{Dold2017} (see also \cite{Cardoso2004}) on black holes which lie below the Hawking--Reall \cite{Hawking2000} bound.

In light of all these mode \textit{instabilities}, it is remarkable that, in the $\Lambda=0$ case, mode \textit{stability} holds for the \textit{entire} Kerr black hole family, even in the endpoint case $|a|=M$.  This is mostly to the credit of the pioneering work of Whiting \cite{Whiting1989}. In 1989, Whiting proved mode stability for $\Im\omega>0$ and $|a|<M$ by demonstrating that such mode solutions \eqref{eq:separable-ansatz-intro} to \eqref{eq:Teukolsky-equation-intro} can be \textit{injectively} mapped into mode solutions of a scalar wave equation on a new spacetime in which the energy associated to the stationary Killing field is coercive. Whiting's map consists of taking appropriate integral and differential transformations of the radial and angular functions, respectively, in \eqref{eq:separable-ansatz-intro}. It turns out, as Shlapentokh-Rothman showed in 2015 \cite{Shlapentokh-Rothman2015}, that Whiting's integral radial transformation suffices to prove mode stability in the upper half-plane and even on the real axis for $|a|<M$, see also a different extension to the real axis in \cite{Andersson2017}. However, as the transformation is very sensible to the change in the nature of singularities in the radial ODE, it breaks down at the endpoint case $|a|=M$. There, Whiting's method requires a very different integral transformation, which was only very recently found by the second author \cite{TeixeiradaCosta2019}. We emphasize that while Whiting's approach to mode stability may be bypassed when considering $|a|\ll M$ black holes, by the reasons described above, it was absolutely crucial to the characterization of solutions to the Teukolsky equation \eqref{eq:Teukolsky-equation} in the \textit{full subextremal range} of parameters $|a|<M$ in the works \cite{Dafermos2016b,SRTdC2020,SRTdC2021}, as no other approach to mode stability for general $a$ was known.

Turning to the Kerr-de Sitter $\Lambda>0$ case, the picture is much less complete, and it remains an open problem to determine whether \eqref{eq:Teukolsky-equation-intro} is modally stable for general black hole parameters. The lack of a mode stability statement is one of the reasons why Hintz and Vasy's proof \cite{Hintz2016} of nonlinear stability of the Kerr-de Sitter family cannot be extended past the very slowly rotating $|a|\ll M,\Lambda$ setting. In fact, in the linear setting of the scalar wave equation, recent work of  Petersen and Vasy \cite{Petersen2021a} has singled out mode stability as the only obstruction to showing decay in the full subextremal range. This state of affairs is also somewhat surprising: the $\Lambda>0$ case in spherical symmetry and perturbations thereof is much better behaved than the $\Lambda=0$ case (compare, for instance, \cite{Dafermos2010} and \cite{Dyatlov2011}), so naively one would expect mode stability to be easier to show in the former case than in the latter. Yet, to employ Whiting's method in Kerr-de Sitter one requires, much like in the $|a|=M, \Lambda=0$ case, a new radial integral transformation, as the nature of the singular points of the radial ODE is very different from the $\Lambda=0$ case. Attempts at finding such a transformation have been unsuccessful, see \cite{Umetsu2000} for a discussion, and no other mechanism of establishing mode stability has been put forth.

The goal of this paper is precisely to revisit mode stability for Kerr(-de Sitter) black holes. In the $\Lambda=0$ setting, we provide a new proof of the classical mode stability result
\begin{bigtheorem} \label{thm:mode-stability-Kerr-intro} Fix $M>0$ and $|a|<M$, and let $s\in\frac12\mathbb{Z}$. Then there are no non-trivial mode solutions  \eqref{eq:separable-ansatz-intro} to \eqref{eq:Teukolsky-equation-intro} for $\omega\neq 0$ with $\Im\omega\geq 0$.
\end{bigtheorem}

Our proof makes use of previously unknown symmetries of the point spectrum of the radial Teukolsky equation. Such symmetries were conjectured to exist by Aminov, Grassi and Hatsuda \cite{Aminov2020,Hatsuda2020} by comparing the Teukolsky equation with quantization conditions for some supersymmetric gauge theories \cite{Ito2017}, see also \cite{Bianchi2021,Bonelli2021}. To establish their conjecture, we rely on a Jaffé expansion for the radial Teukolsky  ODE, see \cite[Part B]{Ronveaux1995}, a method usually attributed in the black hole community to Leaver's seminal work on quasinormal modes \cite{Leaver1985}.  We also sketch an alternative proof via the so-called MST method of Mano, Suzuki and Tagasuki \cite{Mano1996}, see also \cite{Sasaki2003}. Finally, we emphasize that our results hold for $\Im\omega\geq 0$, and we refer the reader to the previous references and the more recent \cite{Gajic2019,Gajic2020} for results concerning the case $\Im\omega<0$.

In the $\Lambda>0$ or Kerr-de Sitter setting, we show that symmetries analogous to those for $\Lambda=0$ hold for the point spectrum of the radial Teukolsky ODE. These are once again inspired by the supersymmetric gauge theories of \cite{Ito2017}, though to the best of our knowledge have not been conjectured or shown earlier. To prove their existence, rather than a Jaffé expansion, we rely on an expansion in hypergeometric polynomials which, despite being well-known in the classical texts on special ODEs, is to our knowledge new in the General Relativity literature. An alternative proof based on a variant of the MST method introduced in \cite{Suzuki1999,Suzuki2000} is also sketched briefly. As for Kerr, making use of these novel symmetries, we are able to establish a partial mode stability result for Kerr-de Sitter:
\begin{bigtheorem} \label{thm:partial-mode-stability-KdS-intro}
Fix $\Lambda>0$,  $M>0$ and $|a|<3/\Lambda$ so that \eqref{eq:Delta-intro} has four distinct real roots, labeled $r_3<r_0<r_1<r_2$, and let $s\in\frac12\mathbb{Z}$. Then there are no non-trivial mode solutions \eqref{eq:separable-ansatz-intro} to \eqref{eq:Teukolsky-equation-intro} with $\omega$ such that 
\begin{gather*}
\omega\in\mathbb{R}\text{~~and~~} m=0 \text{~~or~~} \frac{\omega}{m} \not\in \lp(\frac{2a}{a^2+3/\Lambda-(r_0+r_1)^2}, \frac{2a}{a^2+3/\Lambda-(r_0+r_2)^2}\rp)
\end{gather*}
nor with $\omega$ such that
\begin{gather*}
\Im\omega>0\text{~~and~~}|\omega|\not\in |m|\lp(0, \frac{2a}{a^2+3/\Lambda-(r_0+r_2)^2}\rp)\,.
\end{gather*}
If $|s|=\frac12,\frac32$, in fact there are no non-trivial mode solutions \eqref{eq:separable-ansatz-intro} to \eqref{eq:Teukolsky-equation-intro} for any $\omega\in\mathbb{R}$.
\end{bigtheorem}

Note that no smallness assumptions are made on the black hole specific angular momentum $a$ nor on how close $\omega$ is to the endpoints of the superradiant regime, given in \eqref{eq:superradiance-KdS-intro} for real $\omega$. We also remark that the $\Lambda\to 0$ limit of the proof of Theorem~\ref{thm:partial-mode-stability-KdS-intro} yields precisely Theorem~\ref{thm:mode-stability-Kerr-intro}.

Finally, let us comment on the distinction between cases $|s|=\frac12,\frac32$ and the rest. Similarly to Kerr, if $|s|=\frac12,\frac32$ then the conservation law associated to the stationary Killing is coercive for $\omega\in\mathbb{R}$, and mode stability then follows. In fact, superradiance does not occur for half-integer $s$ in general: the only obstacle to coercivity of the conservation law in this setting is the possible negativity of the so-called Teukolsky--Starobinsky constants for $|s|>2$, which, see our previous work \cite{CasalsTdC2021},  may occur outside the superradiant set \eqref{eq:superradiance-KdS-intro}. Thus, the case $s\in\mathbb{Z}$ is where Theorem~\ref{thm:partial-mode-stability-KdS-intro} is most useful: it rules out \textit{some} of the modes in the superradiant range \eqref{eq:superradiance-intro}, see Figure~\ref{fig:partial-mode-stability-KdS}.

\begin{figure}[htbp]
\centering
\includegraphics[scale=0.8]{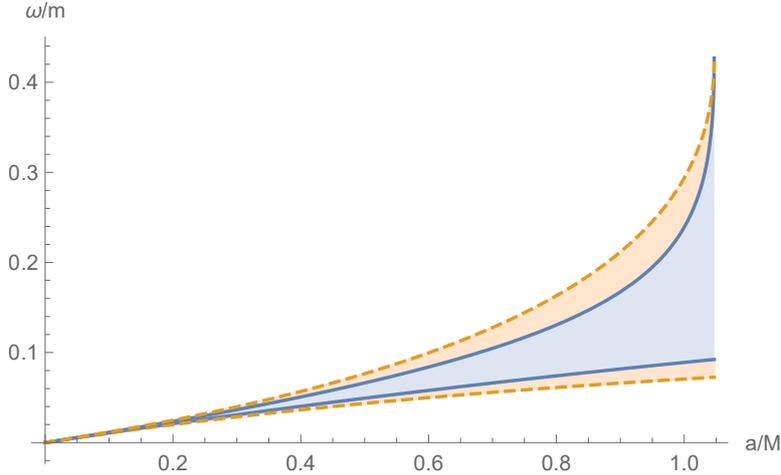}
\caption{The non-superradiant (blank) and superradiant (shaded) regions of separation parameters  for a subextremal Kerr-de Sitter black hole with $\Lambda M^2=1/9$. Theorem~\ref{thm:partial-mode-stability-KdS} rules out the existence of superradiant modes which lie in the orange shaded region between the dashed and full lines. 
}
\label{fig:partial-mode-stability-KdS}
\end{figure}

The remainder of this paper is organized into two sections. Section~\ref{sec:kerr} addresses the case $\Lambda=0$ and gives a new proof of Theorem~\ref{thm:mode-stability-Kerr-intro}. Section~\ref{sec:kds} addresses the case $\Lambda>0$ and contains the proof of  Theorem~\ref{thm:partial-mode-stability-KdS-intro}.

\medskip

\noindent \textbf{Acknowledgments.}  M.C.\ acknowledges partial financial support by CNPq (Brazil), process number 314824/2020-0.
R.TdC.\ acknowledges support from EPSRC (United Kingdom) grant EP/L016516/1 and from NSF (United States) award DMS-2103173, and thanks André Guerra for his enthusiasm for this project and many useful discussions, as well as Mihalis Dafermos for useful suggestions and Igor Rodnianski for insightful comments. Both authors thank one of the anonymous referees for a very careful reading of this manuscript, and many useful suggestions.

\medskip

\noindent \textbf{Data availability statement.} Data sharing not applicable to this article as no datasets were generated or analyzed during the current study.


\section{Subextremal Kerr black holes}
\label{sec:kerr}

\subsection{Geometry of the exterior}
\label{sec:geometry}

In this section, we recall, for the benefit of the reader, some of the basic geometric properties of subextremal Kerr black holes, see for instance \cite{Chandrasekhar} or \cite[Section 5.1]{Dafermos2008} for more details. Fix $M>0$, $|a|<M$, and let
\begin{align*}
r_\pm := M\pm \sqrt{M^2-a^2}\,.
\end{align*}
The subextremal Kerr black hole exterior is a manifold covered globally (modulo the usual degeneration of polar coordinates) by so-called Boyer--Lindquist coordinates $(t,r,\theta,\phi)\in \mathbb{R}\times (r_+,\infty)\times \mathbb{S}^2$ \cite{Boyer1967}, and endowed with the Lorentzian metric
\begin{align*}
g=-\frac{\Delta}{\rho^2}(dt-a\sin^2\theta d\phi)^2+\frac{\rho^2}{\Delta}dr^2+\rho^2d\theta^2+\frac{\sin^2\theta}{\rho^2}\lp(adt-(r^2+a^2)d\phi\rp)^2\,,
\end{align*}
where we have
\begin{align*}
\rho^2:=r^2+a^2\cos^2\theta\,, \qquad\Delta:=r^2-2Mr+a^2=(r-r_+)(r-r_-)\,.
\end{align*}

Finally, we will also find it convenient to work with a rescaling of the Boyer-Lindquist $r$, the tortoise coordinate $r^*=r^*(r)$ defined by
\begin{align*}
\frac{dr^*}{dr}=\frac{r^2+a^2}{\Delta}\,, \quad r^*(3M)=0\,.
\end{align*}
We remark that, throughout this section, we take $'$ to denote a derivative with respect to $r^*$.


\subsection{The Teukolsky equation and its separability}
\label{sec:Teukolsky-separable}

Fix $M>0$, $|a|\leq M$ and $s\in\frac12 \mathbb{Z}$. The Teukolsky equation \cite{Teukolsky1973} is
\begin{equation}
\begin{split}
\Box_g\upalpha^{[s]} &+\frac{2s}{\rho^2}(r-M)\p_r\upalpha^{[s]} +\frac{2s}{\rho^2}\lp[\frac{a(r-M)}{\Delta}+i\frac{\cos\theta}{\sin^2\theta}\rp]\p_\phi\upalpha^{[s]} 
\\
&+\frac{2s}{\rho^2}\lp(\frac{(r-M)(r^2+a^2)}{\Delta}-2r -ia\cos\theta\rp)\p_t\upalpha^{[s]} +\frac{1}{\rho^2}\lp(s-s^2\cot^2\theta\rp)\upalpha^{[s]}  =0\,,
\end{split}\label{eq:Teukolsky-equation}
\end{equation}
in Boyer--Lindquist coordinates, writing $\Box_{g}$ for the covariant wave operator on the Kerr metric $g$. Here, $\swei{\upalpha}{s}$   is a smooth, $s$-spin weighted function on the subextremal Kerr black hole exterior, see \cite[Section 2.2.1]{Dafermos2017} for a precise definition. 

As Teukolsky noted in his seminal paper \cite{Teukolsky1973}, by analogy with the wave equation case \cite{Carter1968}, the Teukolsky equation \eqref{eq:Teukolsky-equation} is separable, i.e.\ it admits separable solutions: 
\begin{align} \label{eq:separable-ansatz}
\swei{\upalpha}{s}(t,r,\theta,\phi)=e^{-i\omega t}e^{im\phi}{S}^{[s],\,a\omega}_{m,\lambdabar}(\theta) \Delta^{-\frac{s+1}{2}}{R}^{[s],\,a, \omega}_{m,\lambdabar}(r)\,,
\end{align}
for $\omega\in\mathbb{C}$, $m-s\in\mathbb{Z}$ and a separation constant $\lambdabar$. Plugging \eqref{eq:separable-ansatz} into \eqref{eq:Teukolsky-equation}, we find that ${S}^{[s],\,a\omega}_{m,\lambdabar}$ and ${R}^{[s],\,a, \omega}_{m,\lambdabar}$ each satisfy ODEs, which are introduced in the next two subsections.


\subsubsection{The angular ODE and its eigenvalues}

Let $s\in\frac12\mathbb{Z}$ be fixed. Consider \eqref{eq:separable-ansatz} and replace $a\omega$ by a parameter $\nu\in\mathbb{C}$. The angular ODE verified by ${S}^{[s],\,\nu}_{m,\lambdabar}$ is 
\begin{gather}
\begin{split}
\frac{1}{\sin\theta}\frac{d}{d\theta}\lp(\sin\theta\frac{d}{d\theta}\rp)S_{m,\lambdabar}^{[s],\,\nu}(\theta)
- \lp(\frac{(m+s\cos\theta)^2}{\sin^2\theta}-\nu^2\cos^2\theta+2\nu s \cos\theta\rp)S_{m,\lambdabar}^{[s],\,\nu}(\theta)+\lambdabar S_{m,\lambdabar}^{[s],\,\nu}(\theta)=0 \,.
\end{split} \label{eq:angular-ode}
\end{gather}
We are interested in solutions of \eqref{eq:angular-ode} with boundary conditions which ensure that, when $\nu$ is taken to be $a\omega$, \eqref{eq:separable-ansatz} is a smooth $s$-spin weighted function on the subextremal Kerr exterior. We quote from \cite[Proposition 2.1]{TeixeiradaCosta2019} a characterization of such solutions, based on the classical references \cite{Meixner1954}, \cite[Section III]{Hartle1974} and \cite[Pages 72--74]{Stewart1975}.

\begin{lemma}[Smooth spin-weighted solutions of the angular ODE] \label{lemma:angular-eigenvalues}
Fix $s\in\frac12\mathbb{Z}$, let $m-s\in\mathbb{Z}$, and assume $\nu\in\mathbb{C}$. Consider the angular ODE \eqref{eq:angular-ode} with the boundary condition that $e^{im\phi}S_{m,\lambdabar}^{[s],\,\nu}$ is a non-trivial smooth $s$-spin-weighted function on $\mathbb{S}^2$, see the precise definition in \cite[Definition 2.2]{TeixeiradaCosta2019}. 

\textbf{The case $\nu\in\mathbb{R}$}. For each $\nu\in\mathbb{R}$, there are countably many such solutions to \eqref{eq:angular-ode} each corresponding to a real value of $\lambdabar$. We index the solutions and eigenvalues by a discrete $l$: $S_{ml}^{[s],\,\nu}$ solves \eqref{eq:angular-ode} with eigenvalue $\lambdabar=\bm\uplambda_{ml}^{[s],\,\nu}$ and induces a complete orthonormal basis, $\{e^{im\phi}S_{ml}^{[s],\,\nu}\}_{ml}$, of the space of smooth $s$-spin-weighted functions on $\mathbb{S}^2$ endowed with $L^2(\sin\theta d\theta)$ norm. The index $l$ is chosen so that $l-\max\{|s|,|m|\}\in\mathbb{Z}_{\geq 0}$, and so that $\bm\uplambda_{ml}^{[s],\,0}=l(l+1)-s^2$ for $\nu =0$ and $\bm\uplambda_{ml}^{[s],\,\nu}$ varies smoothly with $\nu$.  The eigenvalues also have the property that $\bm\uplambda_{ml}^{[s],\,\nu}=\bm\uplambda_{ml}^{[-s],\,\nu}$.

\textbf{The case $\nu\in\mathbb{C}\backslash\mathbb{R}$}. Fix some $\nu_0\in\mathbb{R}$. The corresponding eigenvalue $\bm\uplambda_{ml}^{[s],\,\nu_0}\in \mathbb{R}$ can be analytically continued to $\nu\in\mathbb{C}$ except for finitely many branch points (with no finite accumulation point), located away from the real axis, and branch cuts emanating from these. We define $\bm\uplambda_{ml\nu_0}^{[s],\,\nu}$, for $\nu_0\in\mathbb{R}$, as a global multivalued complex function of $\nu$ such that $\bm\uplambda_{ml\nu_0}^{[s],\,\nu_0}=\bm\uplambda_{ml}^{[s],\,\nu_0}$ and $S_{ml\nu_0}^{[s],\,\nu}$ as a solution to \eqref{eq:angular-ode} with $\lambdabar=\bm\uplambda_{ml\nu_0}^{[s],\,\nu}$. The eigenvalues are independent of $\sign s$ and satisfy
\begin{align}\label{eq:angular-eigenvalues-upper-half-plane}
\Im\nu>0\implies \Im\lp( \overline{\nu}\,\bm\uplambda_{ml\nu_0}^{[s],\,\nu}\rp)<0\,.
\end{align}
\end{lemma}

The reader will find a proof of \eqref{eq:angular-eigenvalues-upper-half-plane} in the positive cosmological setting in Lemma~\ref{lemma:angular-eigenvalues-KdS} below.

\begin{remark} We note that here and throughout the section, $\lambdabar$ denotes a complex number with no restrictions whereas $\bm\uplambda_{ml}^{[s],\,\nu}$ and $\bm\uplambda_{ml\nu_0}^{[s],\,\nu}$ denote one of the eigenvalues identified in Lemma~\ref{lemma:angular-eigenvalues}.
\end{remark}


\subsubsection{The radial ODE and its boundary conditions}

Fix $M>0$. For $|a|\leq M$, $s\in\frac12\mathbb{Z}$, $m-s\in\mathbb{Z}$, $\omega\in \mathbb{C}$ and $\lambdabar\in\mathbb{C}$, the radial ODE verified by $\smlambda{R}{ s}(r)$ in \eqref{eq:separable-ansatz} is, for $r\in(r_+,\infty)$,
\begin{align*}
\Delta\frac{d^2}{dr^2} &\smlambda{R}{ s}(r) +\frac{[\omega(r^2+a^2)-am-is(r-M)]^2}{\Delta}\smlambda{R}{ s}(r) \\
&+\lp(\frac{M^2-a^2}{\Delta}+4is\omega r -\lambdabar-a^2\omega^2 +2am\omega\rp)\smlambda{R}{ s}(r)=0\,,\numberthis \label{eq:radial-ODE-alpha}
\end{align*}
Let us introduce the notation
\begin{gather}
\eta:=  i\frac{\omega-m\upomega_-}{2\kappa_-}\,,  \qquad\xi:= -i\frac{\omega-m\upomega_+}{2\kappa_+}\,, \qquad \upomega_\pm:= \frac{a}{2Mr_\pm}\,,\qquad \kappa_\pm =\frac{r_+-r_-}{4Mr_\pm}\,.
\end{gather}
For $j=\pm$, the quantities $ \upomega_j$ and $\kappa_j$ are, respectively, the angular velocity and the surface gravity of the horizon at $r=r_j$. As before, we are interested in studying \eqref{eq:radial-ODE-alpha} under boundary conditions which ensure that \eqref{eq:separable-ansatz} can arise from suitably regular initial data for the Teukolsky equation~\eqref{eq:Teukolsky-equation}. Based on the classical theory of regular and irregular singularities for ODEs, see \cite[Chapters 5 and 7]{Olver1973}, we will consider the following boundary conditions:

\begin{definition} \label{def:bdry-conditions-Kerr} Assume $\omega\neq 0$. We say a solution, $\smlambda{R}{s}(r)$, to \eqref{eq:radial-ODE-alpha} is 
\begin{itemize}
\item ingoing at $\mc{H}^+$ if the following are smooth at $r=r_+$:
\begin{itemize}
\item $\smlambda{R}{ s}(r)(r-r_+)^{-\xi+\frac{s-1}{2}}$ if either $\Re\omega\neq m\upomega_+$ or $s\leq 0$, and
\item $\smlambda{R}{ s}(r)(r-r_+)^{-\frac{s+1}{2}}$ if  $\Re\omega=m\upomega_+$ and $s\geq 0$;
\end{itemize}
\item outgoing at $\mc{I}^+$ if $\smlambda{R}{ s}(r)e^{-i\omega r}r^{-2iM\omega+s}$  admits an asymptotic series as $r\to \infty$ in powers of $r^{-1}$.
\end{itemize}
\end{definition}


\subsubsection{Precise definition of mode solution}

We are finally ready to define mode solution precisely:
\begin{definition} Fix $M>0$ and $|a|<M$. Take $s\in\frac12\mathbb{Z}$, $m-s\in\mathbb{Z}$ and $\omega\in\mathbb{C}\backslash\{0\}$ with $\Im \omega\geq 0$. Let $\swei{\upalpha}{s}(t,r,\theta,\phi)$ be a solution to \eqref{eq:Teukolsky-equation} which is given by \eqref{eq:separable-ansatz}. We say that $\swei{\upalpha}{s}(t,r,\theta,\phi)$ is a mode solution if
\begin{enumerate}
\item $\lambdabar=\bm\lambda_{ml\nu_0}^{[s],\,\nu}$ for some $l\in\mathbb{Z}_{\geq\max\{|m|,|s|\}}$ and $\nu_0\in\mathbb{R}$, making $e^{im\phi}S_{m,\lambdabar}^{[s],\,a\omega}$  a non-trivial smooth $s$-spin-weighted function on $\mathbb{S}^2$ such that $S_{m,\lambdabar}^{[s],\,a\omega}$ solves the angular ODE \eqref{eq:angular-ode};
\item ${R}^{[s],\,a, \omega}_{m,\lambdabar}(r)$ solves the radial ODE \eqref{eq:radial-ODE-alpha} with ingoing boundary conditions at $\mc{H}^+$ and outgoing boundary conditions at $\mc{I}^+$.
\end{enumerate}
\end{definition}

\begin{remark} \label{rmk:why-mode-solution-def} Let $\swei{\upalpha}{s}(t,r,\theta,\phi)$ be given by \eqref{eq:separable-ansatz} where $S^{[s],(a\omega)}_{m\lambdabar}$ and $\lambdabar$  are one of the functions and eigenvalues identified in Lemma~\ref{lemma:angular-eigenvalues}. If $\smlambda{R}{s}$ is ingoing at $\mc{H}^+$, then $\Delta^{s}\swei{\upalpha}{s}$ is regular at $r=r_+$ \cite[Lemma 2.7]{TeixeiradaCosta2019}. If, furthermore, $\smlambda{R}{s}$ is outgoing at $\mc{I}^+$, then $\swei{\upalpha}{s}$ has finite (weighted, see e.g.\ \cite{Dafermos2017}) energy on suitable spacelike hypersurfaces, see \cite[Appendix D]{Shlapentokh-Rothman2015}. Hence, mode solutions are solutions to \eqref{eq:Teukolsky-equation}  which are separable, regular and finite energy with respect to the algebraically special frame in which the Teukolsky equation is derived.
\end{remark}


\subsection{Some hidden spectral symmetries}
\label{sec:hidden-symmetries}

Consider the radial ODE
\begin{align}
\begin{split}
&\lp[z(z-1)\frac{d^2}{dz^2}-p^2z(z-1)-m_3\, p\, (2z-1)+\lp(E+\frac14\rp)-\frac{m_1 m_2}{z}-\frac{[(m_1+m_2)^2-1]}{4z(z-1)}\rp]y(z)=0\,,
\end{split} \label{eq:Aminov-ODE}
\end{align}
where $m_1,m_2,m_3,E,p\in \mathbb{C}$. It is easy to see that \eqref{eq:radial-ODE-alpha} may be cast in this form:

\begin{lemma}[{\cite[Section 4]{Aminov2020}}] \label{lemma:SQCD-Kerr} The radial ODE \eqref{eq:radial-ODE-alpha} may be rewritten as \eqref{eq:Aminov-ODE} if we choose $z:=\frac{r-r_-}{r_+-r_-}$ and we identify 
\begin{gather}\label{eq:SQCD-masses}
\begin{gathered}
m_1 := s-\xi-\eta=s+ 2iM\omega\,, \qquad m_2:= \eta-\xi=i\frac{2M^2\omega-am}{\sqrt{M^2-a^2}}=\frac{i}{2}\lp(\frac{\omega-m\upomega_+}{\kappa_+}+\frac{\omega-m\upomega_-}{\kappa_-}\rp)\,, \\ m_3:= -s-\xi-\eta=-s+2iM\omega\,, \\
p:= 2i\omega\sqrt{M^2-a^2}\,, \qquad E:= -\lambdabar-a^2\omega^2-s^2+8M^2\omega^2-\frac14\,.
\end{gathered}
\end{gather}
Furthermore, the boundary conditions in Definition~\ref{def:bdry-conditions-Kerr} can be recast in terms of the new parameters, noting that 
\begin{gather*}
-\eta_1+\frac{s-1}{2}=\frac12(m_1+m_2-1)\,, \quad -\frac{s+1}{2}=-\frac12(m_1+m_2+1)\,.
\end{gather*}
Furthermore, we can recast the ODE boundary conditions in terms of the new parameters: $\smlambda{R}{s}$ is ingoing at $\mc{H}^+$ if the functions $\smlambda{R}{s}(r)(r-r_+)^{\frac12(m_1+m_2-1)}$ for $\Re\omega\neq m\upomega_1$, and $\smlambda{R}{s}(r)(r-r_+)^{-\frac12(m_1+m_2+1)}$ otherwhise,  are smooth at $r=r_+$ and outgoing at $\mc{I}^+$ if $\smlambda{R}{s}(r)e^{-p r/(r_+-r_-)}r^{-m_3}$ admits an asymptotic series as $r\to \infty$ in powers of $r^{-1}$.

The choices of $m_1$ and $m_2$ are not canonical: these boundary conditions and the ODE~\eqref{eq:Aminov-ODE} are invariant by an exchange of $m_1$ and $m_2$.
\end{lemma}
\begin{proof}
We note the identities
\begin{align*}
&\omega(r^2+a^2)-am-is(r-M)\\
&\quad= \omega\Delta+\frac{2Mr_+\omega-am}{2}+\frac{2Mr_-\omega-am}{2}+(r-M)(2M\omega-is) \\
&\quad=\omega\Delta-i\lp[\frac12(\eta-\xi)(r_+-r_-)+\frac12(2Mi\omega+s)(r_+-r_-)+(r-r_+)(2Mi\omega+s)\rp]\,,\\
&[\omega(r^2+a^2)-am-is(r-M)]^2 \\
&\quad=\omega^2\Delta^2 -\frac{(\eta-\xi+2Mi\omega+s)^2(r_+-r_-)^2}{4}-(\eta-\xi)(2Mi\omega+s)(r-r_+)(r_+-r_-)\\
&\quad\qquad-\Delta\lp[(2Mi\omega+s)^2+i\omega (\eta-\xi)(r_+-r_-)+2i\omega(2Mi\omega-s+2s)(r-M)\rp]\\
&\quad=\omega^2\Delta^2- 2i\omega(2Mi\omega-s)\Delta(r-M) -\frac{(\eta-\xi+2Mi\omega+s)^2(r_+-r_-)^2}{4}\\
&\quad\qquad -(\eta-\xi)(2Mi\omega+s)(r-r_+)(r_+-r_-)-\Delta\lp[2am\omega-8M^2\omega^2+s^2+4s i\omega r \rp]\,.
\end{align*}
The potential in \eqref{eq:radial-ODE-alpha} can thus be written as
\begin{align*}
&\frac{[\omega(r^2+a^2)-am-is(r-M)]^2+M^2-a^2}{\Delta}+4is\omega r -\lambdabar-a^2\omega^2 +2am\omega\\
&\quad=\omega^2\Delta- 2i\omega(2Mi\omega-s)(r-M) -\frac{[(\eta-\xi+2Mi\omega+s)^2-1](r_+-r_-)^2}{4\Delta} \\
&\quad\qquad-\frac{(\eta-\xi)(2Mi\omega+s)(r-r_+)(r_+-r_-)}{\Delta}-\lambdabar-s^2+8M^2\omega^2-a^2\omega^2\,,
\end{align*}
from where we may now read off the appropriate values of $m_1,m_2,m_3,p,E$.
\end{proof}

\begin{remark}[Connection with SQCD] In \cite{Aminov2020} (see also \cite{Hatsuda2020},  Aminov, Grassi and Hatsuda realized that an analogy can be drawn between the Teukolsky equation on subextremal Kerr black holes and the $SU(2)$ Seiberg--Witten theory with three fundamental hypermultiplets in supersymmetric quantum chromodynamics (SQCD) \cite[Section 3]{Ito2017}: the latter provides a template ODE, \eqref{eq:Aminov-ODE}, which can be matched to the Teukolsky radial ODE \eqref{eq:radial-ODE-alpha} as in Lemma~\ref{lemma:SQCD-Kerr}. We emphasize that the connection to SQCD serves merely as a motivation to write \eqref{eq:Aminov-ODE} and \eqref{eq:SQCD-masses}, and the theory plays no role in the proof of Theorem~\ref{thm:mode-stability-Kerr-intro}.
\end{remark}

Next, we give a characterization of the point spectrum of \eqref{eq:Aminov-ODE}, in the space of solutions with suitable boundary conditions, by a Jaffé expansion \cite{Jaffe1934}, also known in the black hole community as Leaver's method \cite{Leaver1985}.

\begin{proposition}[Point spectrum] \label{prop:Jaffe-expansion} 
Let $E\in\mathbb{C}$ and $p\in\mathbb{C}\backslash\{0\}$ with $\Re p\leq 0$. Consider a triple $(m_{i_1},m_{i_2},m_{i_3})$, where $i_1$, $i_2$ and $i_3$ are distinct natural numbers between 1 and 3, of complex numbers verifying three conditions: $\Re(m_{i_1}+m_{i_2})\leq 0$ if $\Im(m_{i_1}+m_{i_2})=0$, and $\Re(m_{i_1}+m_{i_3})\leq 0$ with $m_{i_1}+m_{i_3}\neq 0$. 

Let ${R}_{(m_{i_1},m_{i_2},m_{i_3})}^{E,p}$ be the unique, up to rescaling, solution of the differential equation (c.f.\ \eqref{eq:Aminov-ODE})
\begin{align*}
\begin{split}
&\lp[z(z-1)\frac{d^2}{dz^2}-p^2z(z-1)-m_{i_3}p(2z-1)+E+\frac14-\frac{m_{i_1} m_{i_2}}{z}-\frac{[(m_{i_1}+m_{i_2})^2-1]}{4z(z-1)}\rp]R(z)=0\,,
\end{split} 
\end{align*}
with boundary conditions
\begin{itemize}
\item ${R}(z)(z-1)^{\frac12(m_{i_1}+m_{i_2}-1)}$ is smooth at $z=1$,
\item ${R}(z)e^{-p z}z^{-m_{i_3}}$ admits an asymptotic series as $z\to \infty$ in powers of $z^{-1}$.
\end{itemize}
Then ${R}_{(m_{i_1},m_{i_2},m_{i_3})}^{E,p}$ is nontrivial if and only if the continued fraction condition
\begin{align}
0= A_0^{(0)}+\frac{-A_0^{(+1)}A_1^{(-1)}}{ A_1^{(0)}+\frac{-A_1^{(+1)}A_2^{(-1)}}{ A_2^{(0)}+\frac{-A_2^{(+1)}A_3^{(-1)}}{\cdots} }} \label{eq:Jaffe-expansion-continued-fraction}
\end{align}
holds, where the coefficients $A^{(0)}_n$ and $A_n^{(\pm 1)}$ satisfy
\begin{align}
\begin{split}
A_{n-1}^{(+1)}A_{n}^{(-1)}&=n\lp\{[n-\sigma_1(\bm{m})] [ n(n-\sigma_1(\bm{m})) +\sigma_2(\bm{m})]+\sigma_3(\bm{m})\rp\}\,,\\
A_n^{(0)}&= E+\frac14+2n(p-n)+ (2n+1-p)\sigma_1(\bm{m})-\sigma_2(\bm{m})\,,
\end{split} \label{eq:Jaffe-expansion-recursion-coefficients}
\end{align}
and where $\sigma_{1}(\bm{m}):= m_1+m_2+m_3$, $\sigma_{2}(\bm{m}):=m_1 m_2+m_1 m_3+m_2m_3$ and $\sigma_{3}(\bm{m}):=m_1 m_2 m_3$ are symmetric polynomials in $\bm{m}=(m_{i_1},m_{i_2},m_{i_3})$.
\end{proposition}

\begin{proof} From the  classical theory of ODEs, see \cite[Chapter 7]{Olver1973}, the second boundary condition is verified if and only if\footnote{It is worth pausing here to note the difference between the cases $\Re p>0$, which is not included in our statement, and $\Re p\leq 0$, which is. In both cases, an analysis of the irregular singularity $z=\infty$ shows that the two linearly independent boundary behaviors are asymptotically proportional to $e^{\pm p z}z^{\pm m_{i_3}}$ as $z\to \infty$; we select the upper sign for the statement. In the case $\Re p\leq  0$, this corresponds to oscillation (if equality holds) or exponential decay as $z\to \infty$, and it can be easily distinguished from the linearly independent behavior which corresponds to oscillation in phase opposition (if equality holds) or exponential growth. However, in the case $\Re p>0$, our preferred boundary condition corresponds to exponential \textit{growth} as $z\to \infty$, which is not easily distinguished from the exponential decay characterizing the linearly independent behavior: both $e^{-p z}z^{- m_{i_3}}\cdot e^{p z}z^{ m_{i_3}}=1$ and $e^{-p z}z^{- m_{i_3}}\cdot e^{-p z}z^{- m_{i_3}}$  have limits as $z\to \infty$. See \cite{Gajic2019,Gajic2020} for a more complete discussion, and proposal for a resolution, of this issue for the case $\Re p>0$.}  we have that  ${R}_{(m_{i_1},m_{i_2},m_{i_3})}^{E,p}(z)e^{-p z}z^{-m_{i_3}}$ has a limit as $z\to \infty$.  We thus conclude that $g$ defined through
\begin{align}
{R}_{(m_{i_1},m_{i_2},m_{i_3})}^{E,p}(z)&=e^{p(z-1)}z^{\frac{1}{2}(m_{i_1}+m_{i_2}+2m_{i_3}-1)}(z-1)^{-\frac{1}{2}(m_{i_1}+m_{i_2}-1)} {g}(z)\,, \label{eq:Jaffe-rescaling}
\end{align}
is the \textit{unique} (up to rescaling) nontrivial solution  to the ODE
\begin{align} \label{eq:Jaffe-ODE}
\begin{split}
&\lp[z(z-1)\frac{d^2}{dz^2}+\lp(B_1z(z-1)+B_2(z-1)+B_3\rp)\frac{d}{dz}+B_4\frac{z-1}{z}+B_5\rp]{g}(z)=0\,,
\end{split}
\end{align} 
with the boundary conditions that $g$ is smooth at $z=1$ and has a limit as $(z-1)/z\to 1$. Here,  the coefficients $B_1,\dots, B_5$ are given by 
\begin{gather} \label{eq:Jaffe-ODE-coefficients}
\begin{gathered}
B_1=2p\,, \quad 
B_2=2m_{i_3}\,, \quad
B_3=1-m_{i_1}-m_{i_2}\,, \quad
B_4=(m_{i_1}+m_{i_3}-1)(m_{i_3}+m_{i_2}-1)\,, \\
B_5=E+\frac14-(m_{i_1}m_{i_2}+m_{i_1}m_{i_3}+m_{i_2}m_{i_3})- (-1 + m_{i_1} + m_{i_2} + m_{i_3}) (p-1)\,.
\end{gathered}
\end{gather}

In light of the strong rigidity in terms of holomorphicity of $g$ around $z=1$ afforded by the classical theory of ODEs, see \cite[Chapter 5]{Olver1973}, it is natural to consider a series expansion in known special functions with such properties at $z=1$. Following \cite[Part B, Chapter 3.2]{Ronveaux1995}, we consider the ansatz for a solution to \eqref{eq:Jaffe-ODE} given by
\begin{align}
\label{eq:Jaffe}
y(z) &= \sum_{n=0}^\infty b_{n}f_n(z)\,, \qquad f_n(z):=\lp(\frac{z-1}{z}\rp)^n\,, 
\end{align}
for some coefficients $\{b_n\}_{n\geq 0}$. Notice that, as long as $y$ can be shown to converge in a neighborhood of $z=1$,  this ansatz is without loss of generality: by uniqueness of solutions to ODEs with prescribed boundary conditions, any solution  of \eqref{eq:Jaffe-ODE} which is regular at $z=1$ must be  a multiple of the ansatz \eqref{eq:Jaffe},  so  $g=Cy$ for some $C\in\mathbb{C}\backslash\{0\}$ in the domain where  $y$ is defined.

Let us turn to determining this domain. From the identities 
\begin{gather*}
\frac{d}{dz}f_n = n\lp(f_{n+1}-2f_n+f_{n-1}\rp)\,, \qquad (z-1)\frac{d}{dz}f_n= n(f_{n}-f_{n+1})\,, \qquad z(z-1)\frac{d}{dz}f_n=nf_n\,,
\end{gather*}
we deduce that the coefficients $\{b_n\}_{n\geq 0}$ in \eqref{eq:Jaffe} satisfy a three-term recurrence formula 
\begin{align}\label{eq:rec rln Kerr}
A_0^{(+1)} b_{1}+A_0^{(0)}b_{0}=0\,;\qquad
A_n^{(+1)}b_{n+1}+A_n^{(0)}b_{n}+A_n^{(-1)}b_{n-1} = 0\,, \quad n\geq 1\,,
\end{align}
with respect to the recurrence coefficients
\begin{gather*}
A_{n+1}^{(-1)}=n(n+1)+(B_3-B_2)n+B_4\,,\qquad
A_{n-1}^{(+1)}=(n-1)n +B_3 n,  \\
A_n^{(0)}=-2n^2 +(B_1+B_2-2B_3)n+B_5\,.
\end{gather*}
The case $b_n\equiv 0$ is trivial. If $b_n\not \equiv 0$, a large-$n$ asymptotic analysis of \eqref{eq:rec rln Kerr} shows that there are  two linearly independent behaviors: writing $C_0:=\sqrt{-B_1}=\sqrt{-2p}$, we have that either
\begin{gather}
\frac{b_{n+1}}{b_n}=1+ \frac{C_0}{n^{1/2}} +O\left(n^{-1}\right)\,,\label{eq:non-minimal-solution-Kerr}
\end{gather}
or 
\begin{align}
\frac{b_{n+1}}{b_n}
=1- \frac{C_0}{n^{1/2}} +O\left(n^{-1}\right)\,. \label{eq:minimal-solution-Kerr}
\end{align}
In either case, $\lim_{n\to\infty}|{b_{n+1}}/{b_n}|=1$, and hence \eqref{eq:Jaffe} converges, and defines a holomorphic function, at least for $|(z-1)/z|<1$. Hence, we deduce that $g\equiv y$ up to rescaling for $|(z-1)/z|<1$.

At this point, let us note that the continued fraction \eqref{eq:Jaffe-expansion-continued-fraction} naturally enters the problem because it is a \textit{formal} solvability condition for \eqref{eq:rec rln Kerr}:
\begin{align*}
\eqref{eq:rec rln Kerr}, \quad b_n\not\equiv 0 &\iff
\begin{dcases}
\frac{b_{n}}{b_{n-1}}=\frac{-A_n^{(-1)}}{A_n^{(0)}+A_n^{(+1)}\frac{b_{n+1}}{b_n}}
=\frac{-A_n^{(-1)}}{A_n^{(0)}+\frac{-A_n^{(+1)}A_{n+1}^{(-1)}}{A_{n+1}^{(0)}+A_{n+1}^{(+1)}\frac{b_{n+2}}{b_{n+1}}}}, &n\geq 1
\\
\frac{b_1}{b_0}=-\frac{A_0^{(0)}}{A_0^{(+1)}}
\end{dcases}\\
&\implies -\frac{A_0^{(0)}}{A_0^{(+1)}}= \frac{-A_1^{(-1)}}{A_1^{(0)}+\frac{-A_1^{(+1)}A_{2}^{(-1)}}{A_{2}^{(0)}+\frac{-A_{2}^{(+1)}A_{3}^{(+1)}}{\cdots}}} \iff \eqref{eq:Jaffe-expansion-continued-fraction}\,,
\end{align*}
which holds \textit{rigorously} as long as  the right hand side of \eqref{eq:Jaffe-expansion-continued-fraction} converges. By \cite[Theorem 1.1]{Gautschi1967}, convergence occurs if and only if the recurrence relation \eqref{eq:rec rln Kerr} possesses a so-called ``minimal solution''. By our assumptions on $p$ we have $\Re C_0>0$, and hence by comparing \eqref{eq:non-minimal-solution-Kerr} and \eqref{eq:minimal-solution-Kerr}, we find that the latter defines a minimal solution to \eqref{eq:rec rln Kerr}. 

We are now ready to conclude. If $g\not\equiv 0$, then $g=y$, up to rescaling, in $|(z-1)/z|<1$, and so $b_n\not\equiv 0$; that \eqref{eq:Jaffe-expansion-continued-fraction} holds then follows directly from \cite[Theorem 1.1]{Gautschi1967}. Conversely, if \eqref{eq:Jaffe-expansion-continued-fraction} holds, then \cite[Theorem 1.1]{Gautschi1967} tells us we can choose $b_n\not\equiv 0$ to satisfy the minimality condition \eqref{eq:minimal-solution-Kerr}. With this choice, our ansatz $y$ in \eqref{eq:Jaffe} not only is well-defined in  $|(z-1)/z|<1$ but also 
has a finite limit, given by $\sum_{n=0}^\infty b_n<\infty$, as $|(z-1)/z|\to 1$. We deduce that $g= y\not\equiv 0$, up to rescaling, in the entire $|(z-1)/z|\leq 1$ domain.
\end{proof}

From the point of view of the Teukolsky equation, Proposition~\ref{prop:Jaffe-expansion} is the statement that its point spectrum, under suitable boundary conditions, is invariant under exchanges $m_i\leftrightarrow m_j$. The precise statement we will use is:

\begin{corollary}[Hidden spectral symmetries] \label{lemma:hidden-symmetry}
Let $M>0$, $|a|<M$, $s\in\frac12\mathbb{Z}$, $m-s\in\mathbb{Z}$, $\lambdabar\in\mathbb{C}$ and $\omega\neq 0$ satisfy $\Im\omega\geq 0$. If $\Re\omega\neq m\upomega_+$, the following are equivalent:
\begin{enumerate}[label=(\roman*)]
\item  (original case and $m_1\leftrightarrow m_2$ symmetry) there is a nontrivial solution to the radial ODE~\eqref{eq:radial-ODE-alpha}  which is ingoing at $\mc{H}^+$  and outgoing at $\mc{I}^+$, i.e.\ which satisfies the boundary conditions
\begin{align*}
&\smlambda{\tilde R}{s}(r)(r-r_+)^{-\xi+\frac{s-1}{2}} \text{~smooth at $r=r_+$}\,, \\ 
&\smlambda{\tilde R}{s}(r)e^{-i\omega r}r^{-2iM\omega +s} \text{~admits an asymptotic series as $r\to \infty$ in powers of $r^{-1}$}\,;
\end{align*}
\item \label{it:case-m1-to-m3}($m_1\leftrightarrow m_3$ symmetry) there is a nontrivial solution to a radial ODE equal to \eqref{eq:radial-ODE-alpha} but where $s$ is replaced by $-s$, which satisfies the boundary conditions
\begin{align*}
&\smlambda{ R}{s}(r)(r-r_+)^{-\xi-\frac{s+1}{2}} \text{~smooth at $r=r_+$}\,, \\ 
&\smlambda{ R}{s}(r)e^{-i\omega r}r^{-2iM\omega-s}\text{~admits an asymptotic series as $r\to \infty$ in powers of $r^{-1}$}\,;
\end{align*}
\item  \label{it:case-m2-to-m3}($m_2\leftrightarrow m_3$ symmetry) there is a nontrivial solution to the radial ODE
\begin{equation}\label{eq:tilded-ODE-R} 
\begin{split}
\Delta\frac{d^2}{dr^2}\smlambda{\tilde R}{s}&-\lp[\lambdabar +2am\omega\frac{r-r_-}{r_+-r_-}
+\omega^2\lp(\Delta +2M^2\frac{r-r_-}{r_+-r_-}-a^2+4M^2\frac{r-r_-}{r-r_+}\rp)
\rp]\smlambda{\tilde R}{s}\\
&+\Delta\lp(\frac{r^2+a^2}{\Delta}\rp)^{\frac12}\frac{d}{dr}\lp[\frac{\Delta}{r^2+a^2}\frac{d}{dr}\lp(\frac{r^2+a^2}{\Delta}\rp)^{\frac12}\rp]\smlambda{\tilde R}{s}\\
&-\lp[\frac{a^2\Delta+2Mr(r^2-a^2)}{(r^2+a^2)^2}+\frac{r-r_+}{r-r_-}s^2\rp]\smlambda{\tilde R}{s}=0
\end{split}
\end{equation}
satisfying the boundary conditions
\begin{equation}\label{eq:case-m2-to-m3-bdry-conditions}
\begin{split}
&\smlambda{\tilde R}{s}(r)(r-r_+)^{2iM\omega-\frac12} \text{~smooth at $r=r_+$}\,, \\ 
&\smlambda{\tilde R}{s}(r)e^{-i\omega r}r^{\xi-\eta} \text{~admits an asymptotic series as $r\to \infty$ in powers of $r^{-1}$}\,.
\end{split}
\end{equation}
\end{enumerate}
If $\Re\omega=m\upomega_+$, then (i) and (iii) are equivalent if $s\leq 0$ and, if $s\geq 0$, (ii) and (iii) are equivalent.
\end{corollary}
\begin{proof} The conclusion follows from Proposition~\ref{prop:Jaffe-expansion} after setting $z(r_+-r_-)=r-r_-$ and rewriting the ODE \eqref{eq:tilded-ODE-R} and the boundary conditions in terms of $(m_1,m_2,m_3,E,p)$.
\end{proof}

A few remarks concerning the above symmetries are in order.

\begin{remark}[Spin-reversal symmetry] \label{rmk:spin-reversal-symmetry} Note that the exchange $m_1\leftrightarrow m_3$ corresponds to the map $s\mapsto -s$. (This is still the case if $\lambdabar$ is taken to be  one of the angular eigenvalues given by Lemma~\ref{lemma:angular-eigenvalues}, as the latter are independent of the sign of $s$.) Thus,  Corollary~\ref{lemma:hidden-symmetry}\ref{it:case-m1-to-m3}, reflects the fact that a mode solution with spin $-s$ exists if and only if a mode solution with respect to the same frequency and black hole parameters exists with spin $+s$. The latter is a well-known statement in the literature which follows from the Teukolsky--Starobinsky identities introduced in \cite{Teukolsky1974,Starobinsky1974}, see e.g.\ \cite{Whiting1989}, and we refer the reader to \cite[Lemma 2.19]{TeixeiradaCosta2019} for an explicit proof of the implication.
\end{remark}

\begin{remark}[Whiting's transformation] Let $\smlambda{\tilde{R}}{s}$ be a solution to \eqref{eq:tilded-ODE-R}; setting 
\begin{equation}
\tilde{u}=(r^2+a^2)^{1/2}/\Delta^{1/2}\smlambda{\tilde{R}}{s}\,, \label{eq:Whiting-tilde-u}
\end{equation} 
we find that the ODE for the $\tilde{u}$ quantity is 
$\tilde{u}''+\tilde{V}\tilde{u}=0$, where we recall  that prime denotes a derivative with respect to the tortoise coordinate $r^*$ and where
\begin{align}
\begin{split}
\tilde{V}&=\frac{\Delta}{(r^2+a^2)^2}\lp[-\lambdabar -2am\omega\frac{2(r-M)}{r_+-r_-}
+\omega^2\lp(4M^2\frac{2(r-M)}{r-r_+}-a^2+4Mr_-\frac{2(r-M)}{r_+-r_-}\rp)
\rp]\\
&\qquad-\frac{\Delta}{(r^2+a^2)^2}\lp[\frac{a^2\Delta+2Mr(r^2-a^2)}{(r^2+a^2)^2}+\frac{r-r_+}{r-r_-}s^2\rp]+\frac{\Delta}{(r^2+a^2)^2}\omega^2\lp(\Delta+4M(r-M)\rp)\,, \label{eq:Whiting-ODE-potential}
\end{split}
\end{align}
as noted by Hatsuda \cite{Hatsuda2020}. Under the assumption $|a|<M$, it has been known  for some time that this ODE has the same (empty) point spectrum as \eqref{eq:radial-ODE-alpha} for $\Im\omega\geq 0$ and $\omega\neq 0$: indeed, this is exactly the ODE obtained by Whiting's (injective) radial transformation in the seminal mode stability paper \cite{Whiting1989} for $\Im\omega>0$ which is also used to show mode stability for $\omega\in\mathbb{C}$ \cite{Shlapentokh-Rothman2015}, see also \cite{Andersson2017} and \cite[Section 3.2]{TeixeiradaCosta2019}.
\end{remark}

\begin{remark}[The extremal limit] Lemma~\ref{lemma:SQCD-Kerr} clearly fails to hold for $|a|=M$, as the variable $z$ and the parameter $m_2$ are not well-defined in this case. The latter degeneration is remedied if we exchange $m_2\leftrightarrow m_3$, as from \eqref{eq:SQCD-masses} we obtain
\begin{gather*}
m_1 =s+ 2iM\omega\,, \qquad p\, m_2=-2\omega(2M^2\omega-am)\,, \qquad 
m_3=-s+2iM\omega\,, \\ p^2=0\,, \qquad E= -\lambdabar-s^2+7M^2\omega^2-\frac14\,.
\end{gather*}
The former degeneration is also cured if we take $z=\frac{z-\tilde r_+}{\tilde r_+-M}$ for some $\tilde r_+>r_+=r_-=M$. 

Let $\smlambda{\tilde{R}}{s}$ be a solution to \eqref{eq:Aminov-ODE} under these conditions and let $'$ denote a derivative with respect to a modified $r^*$ variable defined by
\begin{align*}
\frac{dr^*}{dr}= \frac{r^2+\tilde{r}_+r_-}{(r-r_-)(r-\tilde{r}_+)}
\end{align*}
and an initial condition. Then, the rescaling 
\begin{equation}
\tilde{u}=\frac{(r^2+\tilde{r}_+r_-)^{1/2}}{[(r-r_-)(r-\tilde{r}_+)]^{1/2}}\smlambda{\tilde{R}}{s} \label{eq:extremal-transf-tilde-u}
\end{equation} 
solves the ODE given by $u''+\tilde{V} \tilde{u}=0$, where
\begin{align}
\begin{split}
\tilde{V}&=\frac{(r-\tilde r_+)(r-M)}{(r^2+M\tilde r_+)^2}\lp[-\lambdabar -2am\omega\frac{2r-\tilde r_+-M}{\tilde r_+-M}\rp]\\
&\qquad+\frac{(r-\tilde r_+)(r-M)}{(r^2+M\tilde r_+)^2}\omega^2\lp(4M^2\frac{2r-\tilde r_+-M}{r-\tilde r_+}-M^2+4M^2\frac{2r-\tilde r_+-M}{\tilde r_+-M}\rp)\\
&\qquad-\frac{(r-\tilde r_+)(r-M)}{(r^2+M\tilde r_+)^2}\lp[\frac{M \tilde r_+(r-\tilde r_+)(r-M)+(\tilde r_++M)r(r^2-M\tilde r_+)}{(r^2+M\tilde r_+)^2}+\frac{r-\tilde r_+}{r-M}s^2\rp]\,, \label{eq:extremal-transf-ODE-potential}
\end{split}
\end{align}
as noted by Hatsuda  \cite{Hatsuda2020}. The second author has shown in earlier work that this ODE has the same (empty) point spectrum as \eqref{eq:radial-ODE-alpha} with $|a|=M$ when $\Im\omega\geq 0$ and $\omega\neq 0,m\upomega_+$ by means of an (injective) radial transformation \cite[Section 3.1]{TeixeiradaCosta2019}. Thus, we fully expect that an analogue of Proposition~\ref{prop:Jaffe-expansion} is possible in the doubly confluent Heun, or extremal $|a|=M$ Kerr, case. The additional difficulty one must contend with in the proof is that swapping certain masses changes the nature of the singularities in the ODE: for instance, \eqref{eq:extremal-transf-ODE-potential} is of reduced confluent Heun type whereas \eqref{eq:radial-ODE-alpha}, in the case $|a|=M$, is of doubly confluent Heun type. We refer the reader to \cite{TeixeiradaCosta2019} or the book \cite{Ronveaux1995} for a precise definition of reduced and doubly confluent Heun ODEs and further properties of such equations.
\end{remark}


\subsection{Proof of mode stability}
\label{sec:proof}

In this section, we give a proof of Theorem~\ref{thm:mode-stability-Kerr-intro}, alternative  to those of the literature \cite{Whiting1989,Shlapentokh-Rothman2015, Andersson2017,TeixeiradaCosta2019}, which relies on the hidden symmetries uncovered in Proposition~\ref{prop:Jaffe-expansion} and Corollary~\ref{lemma:hidden-symmetry}. To be precise, we prove
\begin{theorem}\label{thm:mode-stability-Kerr} Fix $M>0$, $|a|< M$, $s\in\frac12\mathbb{Z}$, $m-s\in\mathbb{Z}$ and $(\omega,\lambdabar)$ such that we have either 
$$\Im\omega>0\text{~and~}\Im(\lambdabar\overline{\omega})\leq 0\quad\text{or}\quad\omega\in\mathbb{R}\backslash\{0\}\text{~and~}\lambdabar\in\mathbb{R}\,.$$
If $\smlambda{R}{s}(r)$ is a solution to \eqref{eq:radial-ODE-alpha} with respect to these parameters which is ingoing at $\mc H^+$ and outgoing at $\mc I^+$, then $\smlambda{R}{s}\equiv 0$.
\end{theorem}

\begin{remark} Note that, by Lemma~\ref{lemma:angular-eigenvalues}, it is clear that the conditions on $\lambdabar$ which we impose here are verified by the angular eigenvalues associated to mode solutions. Hence, Theorem~\ref{thm:mode-stability-Kerr} implies, and is stronger than, Theorem~\ref{thm:mode-stability-Kerr-intro}.
\end{remark}

\begin{proof}[Proof of Theorem~\ref{thm:mode-stability-Kerr}] Take $|a|<M$, let $\smlambda{\tilde{R}}{s}$ be a solution to \eqref{eq:tilded-ODE-R} with boundary conditions \eqref{eq:case-m2-to-m3-bdry-conditions} and define $\tilde{u}$ by \eqref{eq:Whiting-tilde-u}. Now consider the microlocal energy current associated to the stationary Killing field on Kerr:
\begin{align}
Q^T[\tilde u]:= \Im \lp(\overline{\omega \tilde u }\tilde u'\rp)\Rightarrow -(Q^T[u])'=\Im\omega |\tilde u'|^2+\Im\lp(\overline{\omega} \tilde V\rp)|\tilde{u}|^2\,, \label{eq:QT-mode-stability}
\end{align}
where prime again denotes a derivative with respect to $r^*$. Note that  \eqref{eq:Whiting-ODE-potential} is real if $\omega$ is real, that the coefficient of $\lambdabar$ in \eqref{eq:Whiting-ODE-potential} and the $(\omega,m,\lambdabar)$-independent part of $\tilde{V}$ are non-positive. 
Thus, from \eqref{eq:QT-mode-stability} and our assumptions, we obtain
\begin{gather}\label{eq:mode-stability-intermediate}
\begin{gathered}
\text{if~}\Im\omega>0\,,\quad 0\geq \Im\omega \int_{-\infty}^\infty \lp( |\tilde{u}'|^2+|\omega|^2\frac{\Delta}{(r^2+a^2)^2}|\tilde{u}|^2\rp)dr^*\,;\\
\text{if~}\omega\in\mathbb{R}\backslash\{0\}\,,\quad \frac{r_+-r_-}{r_+}\omega^2|\tilde{u}(-\infty)|^2+ \omega^2|\tilde{u}(+\infty)|^2= 0\,,
\end{gathered}
\end{gather}

By a direct argument if $\Im\omega>0$ or, if $\omega\in\mathbb{R}\backslash\{0\}$, by unique continuation for solutions to ODEs with the prescribed boundary conditions, see e.g.\ \cite[Lemma 5.1]{Shlapentokh-Rothman2015} and \cite[Lemma 4.1]{TeixeiradaCosta2019}, \eqref{eq:mode-stability-intermediate} implies that $\tilde{u}\equiv 0$ and so $\smlambda{\tilde R}{s}(r)\equiv 0$. By Corollary~\ref{lemma:hidden-symmetry}, using case (i) if $s\leq 0$ and case (ii) if $s\geq 0$, we must also have that $\smlambda{R}{s}(r)\equiv 0$.
\end{proof}


\subsection{Epilogue: mass symmetries within the  MST method}

In this section, we provide an alternative proof of Corollary~\ref{lemma:hidden-symmetry} which is based not on Jaffé expansions of Proposition~\ref{prop:Jaffe-expansion} but on the matching of (confluent) hypergeometric expansions. In the study of quasinormal modes, this method was first introduced by Mano, Suzuki and Tagasuki \cite{Mano1996}, and is thus  known as the MST method. In the exposition below we follow the review \cite{Sasaki2003}.

Let us fix $M>0$, $|a|<M$, $s\in\frac12\mathbb{Z}$, $m-s\in\mathbb{Z}$, $\omega\in\mathbb{C}\backslash\{0\}$ with $\Re\omega\geq 0$ and $\Im\omega\geq 0$, and $\lambdabar\in\mathbb{C}$, assuming additionally that $s\leq 0$ if $\Re\omega=m\upomega_+$. To aid the reader, we write the MST quantities $\epsilon$, $\tau$ and $\kappa$ in \cite{Sasaki2003} in terms of the quantities identified in \eqref{eq:SQCD-masses} and vice-versa:
\begin{gather*}
\epsilon:= 2M\omega=-i\frac{m_1+m_3}{2}\,,\quad 
\tau:= \frac{2M^2\omega-am}{\sqrt{M^2-a^2}}=-im_2\,,\quad 
\kappa:= \frac{\sqrt{M^2-a^2}}{M}=-i\frac{p}{\epsilon}\,,
\\
\epsilon_+:= \frac{\epsilon+\tau}{2}=-\frac{i}{4}\left(m_1+2m_2+m_3\right),
\quad 
m_1=s+i\epsilon\,, \quad 
m_3=-s+i\epsilon\,, \quad s=\frac{m_1-m_3}{2}.
\end{gather*}
Within the MST formalism, existence of a non-trivial solution to \eqref{eq:radial-ODE-alpha} with outgoing boundary conditions at $\mc{I}^+$ and ingoing boundary conditions at $\mc{H}^+$ is, by \cite[Equations (165), (167) and (168)]{Sasaki2003} and the properties of the Gamma function, equivalent to the condition\footnote{When comparing to the references given, the reader may find that the factor ${\Gamma\left(1-m_1-m_2\right)}$ is missing from the denominator of \eqref{eq:MST-qnm-condition}. This is because  
$$1-m_1-m_2=1-s+2i\frac{am-2Mr_+\Re\omega}{r_+-r_-}+\frac{4M r_+}{r_+-r_-}\Im\omega$$
cannot be a negative integer under our assumptions, so that $|\Gamma(1-m_1-m_2)|<\infty$.
} 
\begin{align}
\frac{B}{\Gamma\left(\nu+1+m_3\right)}\lp( \frac{A_\nu C_\nu}{D_\nu}-ie^{-i\pi \nu}\frac{A_{-\nu-1}C_{-\nu-1}}{D_{-\nu-1}}\rp)=0\,, \label{eq:MST-qnm-condition}
\end{align}
where we have used the shorthand notation 
\begin{align*}
A_\nu&:=\left(-2ip\right)^{-\nu}\frac{\Gamma\left(2\nu+2\right)}{\Gamma\left(\nu+1+m_1 \right)\Gamma\left(\nu+1+m_2 \right)\Gamma\left(\nu+1+m_3 \right)}\,,
\\ 
B&:=e^{p}\left(-2ip\right)^{(m_3-m_1)/2}2^{(m_1-m_3)/2}\,,\\
C_\nu&:=\sum_{n=0}^{\infty}
\frac{(-1)^n\Gamma\left(n+2\nu+1 \right)}{n!}\frac{b_n^{\nu}}{\Gamma\left(1+N-m_1\right)\Gamma\left(1+N-m_2\right)\Gamma\left(1+N-m_3\right)}\,,
 \\ 
D_\nu&:=\sum_{n=-\infty}^{0}
\frac{(-1)^n}{(-n)! \left(2\nu+2\right)_n}\frac{b_n^{\nu}}{\Gamma\left(1+N+m_1\right)\Gamma\left(1+N+m_2\right)\Gamma\left(1+N+m_3\right)}\,,
\end{align*}
writing $N:=n+\nu$ and where the coefficients $b_n^\nu$, obtained from the $a_n^\nu$ coefficients of \cite[Equation (120)]{Sasaki2003} through $b_n^{\nu}:=\Gamma\left(1+m_1+N\right)\Gamma\left(1+m_2+N\right)\Gamma\left(1-m_3+N\right)a_n^{\nu}$, satisfy the following recursive relation, c.f.\ \cite[Equations (123, 124)]{Sasaki2003}:
\begin{gather*}
\alpha_n^\nu b_{n+1}^\nu+\beta_n^{\nu}b_n^{\nu}+\gamma_n b_{n-1}^\nu=0\,, \\ 
\alpha_n^\nu:=\frac{p}{(1+N)(3+2N)}\,, \qquad \gamma_n^\nu:=-
\frac{p\left(\prod_{j=1}^3(N^2-m_j^2)\right)
}{N(-1+2N)}\,, \\
\beta_n^{\nu}:=\frac{1}{4}+E+N(N+1)+p\frac{m_1m_2m_3}{N(1+N)}\,.
\end{gather*}
Furthermore, the parameter $\nu$ is the so-called renormalized angular momentum parameter and its value is chosen to ensure that the continued fraction equation
\begin{align}
\frac{\alpha_{n-1}^\nu\gamma_{n}^\nu}{\lp(\beta_n^\nu -\frac{\alpha_n^\nu\gamma_{n+1}^\nu}{\beta_{n+1}^\nu-\cdots}\rp)\lp(\beta_{n-1}^\nu -\frac{\alpha_{n-2}^\nu\gamma_{n-1}^\nu}{\beta_{n-2}^\nu-\cdots}\rp)}=1\label{eq:MST-nu-equation}
\end{align}
holds for an arbitrary choice of $n\in\mathbb{Z}$, c.f.\ \cite[Equation (133)]{Sasaki2003}, and that $1+\nu+m_3$ is not a non-positive integer. Indeed, these two conditions are not mutually exclusive: the condition \eqref{eq:MST-nu-equation} is invariant by translation in $\mathbb{Z}$, i.e.\ if $\nu=\nu_0$ satisfies \eqref{eq:MST-nu-equation} then so does $\nu=\nu_0+k$ for any $k\in\mathbb{Z}$.  Noting that $B\neq 0$, this choice of $\nu$ ensures that \eqref{eq:MST-qnm-condition} is equivalent to 
\begin{align}
\frac{A_\nu C_\nu}{D_\nu}-ie^{-i\pi \nu}\frac{A_{-\nu-1}C_{-\nu-1}}{D_{-\nu-1}}=0\,. \label{eq:MST-qnm-condition-simplified}
\end{align}

Clearly, $\alpha_{n}^\nu$, $\beta_{n}^\nu$ and $\gamma_n^\nu$ are all separately invariant under the map $(m_1,m_2,m_3)\mapsto(m_i,m_j,m_k)$ with $i\neq j\neq k$, and so $b_n^\nu$ and $\nu$ must also be preserved. Consequently, the same is true for $A_\nu$, $C_\nu$ and $D_\nu$. Hence, we conclude  that the condition \eqref{eq:MST-qnm-condition-simplified} is invariant under the map $(m_1,m_2,m_3)\mapsto(m_{i_1},m_{i_2},m_{i_3})$ with $i_1\neq i_2\neq i_3\neq i_1$. By using the fact that \eqref{eq:radial-ODE-alpha} and the boundary conditions are invariant under taking at once $\Re\omega\mapsto-\Re\omega, m\mapsto -m$ and complex conjugation, we arrive at the same conclusion for $\Re\omega\leq 0$. The statement of Corollary~\ref{lemma:hidden-symmetry} then follows easily from exploiting these symmetries.


\section{Subextremal Kerr-de Sitter black holes}
\label{sec:kds}

\subsection{Geometry of the exterior}
\label{sec:geometry-KdS}

In this section, we recall for the benefit of the reader some of the basic geometric properties of subextremal Kerr-de Sitter black holes, see for instance \cite[Section 1]{Dyatlov2011a} for more details. 

Fix $M>0$, $\Lambda>0$ and $|a|<3/\Lambda$ satisfying
\begin{align}
27 M^4 + L^4 (-1 + \Xi) \Xi^4 + 
 L^2 M^2 (-2 + \Xi) (-32 + \Xi (32 + \Xi))<0\, , \label{eq:subextremal-condition-KdS}
\end{align}
where here and throughout the section we use the notation
\begin{equation}
L^2:=\frac{3}{\Lambda}\,, \qquad \Xi:=1+\frac{a^2}{L^2}\,,
\end{equation}
Then, the following quartic function
\begin{align}
\Delta:=(r^2+a^2)\lp(1-\frac{r^2}{L^2}\rp)-2Mr=-\frac{1}{L^2}(r-r_2)(r-r_1)(r-r_0)(r-{r}_3)\,. \label{eq:kerr-dS-delta}
\end{align}
has four distinct real roots, denoted $0<r_0<r_1<r_2$ and $r_3=-r_2-r_1-r_0<0$.

The subextremal Kerr black hole exterior is a manifold covered globally (modulo the usual degeneration of polar coordinates) by so-called Chambers--Moss coordinates $(t,r,\theta,\phi)\in \mathbb{R}\times (r_1,r_2)\times \mathbb{S}^2$ \cite{Chambers1994} and endowed with the Lorentzian metric
\begin{align*}
g=-\frac{\Delta}{\Xi^2\rho^2}(dt-a\sin^2\theta d\phi)^2+\frac{\rho^2}{\Delta}dr^2+\frac{\rho^2}{\Delta_\theta}d\theta^2+\frac{\Delta_\theta\sin^2\theta}{\Xi^2\rho^2}\lp(adt-(r^2+a^2)d\phi\rp)^2\,,
\end{align*}
where we have
\begin{align*}
\rho^2:=r^2+a^2\cos^2\theta\,, \qquad \Delta_\theta=\sin^2\theta + \Xi\cos^2\theta\,.
\end{align*}

Finally, we will also find it convenient to work with a rescaling of the Chambers--Moss $r$, the tortoise coordinate $r^*=r^*(r)$ defined by
\begin{align*}
\frac{dr^*}{dr}=\frac{\Xi(r^2+a^2)}{\Delta}\,,
\end{align*}
and an initial condition. Note that we have
$$\lim_{r\to r_1}r^*(r)=-\infty\,, \qquad \lim_{r\to r_2}r^*(r)=+\infty\,.$$
Furthermore, we remark that in this section derivatives denoted by $'$ will be assumed to be taken with respect to $r^*$, unless otherwise stated.

\subsection{The Teukolsky equation and its separability}
\label{sec:Teukolsky-separable-KdS}

Fix $M>0$, and $|a|,L>0$  subextremal Kerr-de Sitter parameters. For $s\in\frac12 \mathbb{Z}$, the Teukolsky equation in Kerr-de Sitter \cite{Khanal1983} is 
\begin{equation}
\begin{split}
\Box_g\upalpha^{[s]} &+\frac{s}{\rho^2\Xi^2}\frac{d\Delta}{dr}\p_r\upalpha^{[s]} +\frac{2s}{\rho^2\Xi}\lp[\Xi\frac{a}{2\Delta}\frac{d\Delta}{dr}+i\frac{\cos\theta}{\sin^2\theta}-\frac{ia^2}{L^2\Delta_\theta}\cos\theta\rp]\p_\phi\upalpha^{[s]} 
\\
&+\frac{2s}{\rho^2}\lp(\frac{(r^2+a^2)}{2\Delta}\frac{d\Delta}{dr}-2r -ia\frac{\cos\theta}{\Delta_\theta}\rp)\p_t\upalpha^{[s]} \\
&+\frac{s}{\rho^2\Xi^2}\lp[1-\frac{a^2}{L^2}-\frac{2(3+2s)r^2}{L^2}-\Xi^2 s\frac{\cot^2\theta}{\Delta_\theta}\rp]\upalpha^{[s]} =\frac{2\mu}{L^2}\upalpha^{[s]}\,,
\end{split}\label{eq:Teukolsky-equation-KdS}
\end{equation}
in Chambers--Moss coordinates, writing $\Box_{g}$ for the covariant wave operator in the Kerr-de Sitter metric. Here, $\swei{\upalpha}{s}$ is a smooth, $s$-spin weighted function on the exterior, $\mc{R}$, of the Kerr-de Sitter black hole. The parameter $\mu$ must be taken to be 1 if $s\neq 0$, for instance in the especially interesting setting of gravitational perturbations, $s=\pm 2$, of Kerr-de Sitter black holes. In the scalar case $s=0$, $\mu$ represents a dimensionless Klein--Gordon mass, interpolating been the massless wave equation, corresponding to $\mu=0$, and the conformal Klein--Gordon equation, corresponding to $\mu=1$.

By analogy with the wave equation case \cite{Carter1968}, the Teukolsky equation \eqref{eq:Teukolsky-equation-KdS} in Kerr-de Sitter backgrounds is separable, i.e.\ it admits separable solutions: 
\begin{align} \label{eq:separable-ansatz-KdS}
\swei{\upalpha}{s}(t,r,\theta,\phi)=e^{-i\omega t}e^{im\phi}{S}^{[s],\,\Xi,a\omega}_{m,\lambdabar}(\theta)(r-r_3)^{-1}\Delta^{-\frac{s+1}{2}}{R}^{[s],\,\Xi, a, \omega}_{m,\lambdabar}(r)\,,
\end{align}
for $\omega\in\mathbb{C}$, $m-s\in\mathbb{Z}$ and a separation constant $\lambdabar$. Plugging  \eqref{eq:separable-ansatz-KdS} into \eqref{eq:Teukolsky-equation-KdS}, we find that ${S}^{[s],\,\Xi,a\omega}_{m,\lambdabar}$ and ${R}^{[s],\,\Xi,a, \omega}_{m,\lambdabar}$ each satisfy ODEs, which are introduced in the next two subsections.

\subsubsection{The angular ODE and its eigenvalues}

Let $s\in\mathbb{Z}$ be fixed. Consider \eqref{eq:separable-ansatz-KdS} and replace $a\omega$ by a parameter $\nu\in\mathbb{C}$. The angular ODE verified by ${S}^{[s],\,\Xi,\nu}_{m,\lambdabar}$ is
\begin{gather}
\begin{split}
&\frac{1}{\sin\theta}\frac{d}{d\theta}\lp(\Delta_\theta\sin\theta\frac{d}{d\theta}\rp)S_{m,\lambdabar}^{[s],\,\Xi,\nu}(\theta)-2(\Xi-1)\cos^2\theta\lp[(\mu+2s^2)+\frac{\Xi^2}{\Delta_\theta} m\nu\rp]S_{m,\lambdabar}^{[s],\,\nu}(\theta)+\lambdabar S_{m,\lambdabar}^{[s],\,\Xi,\nu}(\theta)\\
&\qquad-\lp[\frac{\lp(\Xi m+s\cos\theta\lp(\Xi-2(\Xi-1)\sin^2\theta\rp)\rp)^2}{\sin^2\theta\Delta_\theta}-\Xi^2\nu^2\frac{\Xi\cos^2\theta}{\Delta_\theta}+2\nu s\cos\theta\frac{\Xi^2}{\Delta_\theta}\rp]S_{m,\lambdabar}^{[s],\,\Xi,\nu}(\theta)=0 \,.
\end{split} \label{eq:angular-ode-KdS}
\end{gather}
We are interested in solutions of \eqref{eq:angular-ode-KdS} with boundary conditions which ensure that, when $\nu$ is taken to be $a\omega$, \eqref{eq:separable-ansatz-KdS} is a smooth $s$-spin weighted function on $\mc{R}$, which can be defined similarly to the Kerr case of the previous section. By analogy with the Kerr case of Lemma~\ref{lemma:angular-eigenvalues}, we have:

\begin{lemma}[Smooth spin-weighted solutions of the angular ODE] \label{lemma:angular-eigenvalues-KdS}
Fix $s\in\frac12\mathbb{Z}$, let $m-s\in\mathbb{Z}$, and assume $\nu\in\mathbb{C}$. Consider the angular ODE \eqref{eq:angular-ode} with the boundary condition that $e^{im\phi}S_{m,\lambdabar}^{[s],\,\Xi,\nu}$ is a non-trivial smooth $s$-spin-weighted function on $\mathbb{S}^2$, i.e.\  $S_{m,\lambdabar}^{[s],\,\Xi,\nu}(\theta)(1+\cos\theta)^{-\frac{|m-s|}{2}}$ is holomorphic in $\theta\in(0,\pi]$, and $S_{m,\lambdabar}^{[s],\,\Xi,\nu}(\theta)(1-\cos\theta)^{-\frac{|m+s|}{2}}$ is holomorphic in $\theta\in[0,\pi)$.

\textbf{The case $\nu\in\mathbb{R}$}. For each $\nu\in\mathbb{R}$, there are countably many such solutions to \eqref{eq:angular-ode-KdS} each corresponding to a real value of $\lambdabar$. We index the solutions and eigenvalues by a discrete $l$: $S_{ml}^{[s],\,\Xi,\nu}$ solves \eqref{eq:angular-ode-KdS} with eigenvalue $\lambdabar=\bm\uplambda_{ml}^{[s],\,\Xi, \nu}$ and induces a complete orthonormal basis, $\{e^{im\phi}S_{ml}^{[s],\,\Xi,\nu}\}_{ml}$, of the space of smooth $s$-spin-weighted functions on $\mathbb{S}^2$  endowed with $L^2(\sin\theta d\theta)$ norm. The index $l$ is chosen so that $l-\max\{|s|,|m|\}\in\mathbb{Z}_{\geq 0}$, and so that $\bm\uplambda_{ml}^{[s],\,\Xi,0}=l(l+1)-s^2$ for $\nu =0$ and $\bm\uplambda_{ml}^{[s],\,\Xi,\nu}$ varies smoothly with $\nu$. The eigenvalues also have the property that $\bm\uplambda_{ml}^{[s],\,\Xi,\nu}=\bm\uplambda_{ml}^{[-s],\,\Xi,\nu}$.

\textbf{The case $\nu\in\mathbb{C}\backslash\mathbb{R}$}. Fix some $\nu_0\in\mathbb{R}$. The corresponding eigenvalue $\bm\uplambda_{ml}^{[s],\,\Xi,\nu_0}\in \mathbb{R}$ can be analytically continued to $\nu\in\mathbb{C}$ except for finitely many branch points (with no finite accumulation point), located away from the real axis, and branch cuts emanating from these. We define $\bm\uplambda_{ml\nu_0}^{[s],\,\Xi,\nu}$, for $\nu_0\in\mathbb{R}$, as a global multivalued complex function of $\nu$ such that $\bm\uplambda_{ml\nu_0}^{[s],\,\Xi, \nu_0}=\bm\uplambda_{ml}^{[s],\,\Xi,\nu_0}$ and $S_{ml\nu_0}^{[s],\,\Xi,\nu}$ as a solution to \eqref{eq:angular-ode-KdS} with $\lambdabar=\bm\uplambda_{ml\nu_0}^{[s],\,\Xi,\nu}$. The eigenvalues are independent of $\sign s$ and satisfy
\begin{align}\label{eq:angular-eigenvalues-upper-half-plane-KdS}
\Im\nu>0\implies \Im\lp( \overline{\nu}\,\bm\uplambda_{ml\nu_0}^{[s],\,\Xi,\nu}\rp)<0\,.
\end{align}
\end{lemma}
\begin{proof}
Follows by adapting the arguments in \cite[Proposition 2.1]{TeixeiradaCosta2019}. For instance, \eqref{eq:angular-eigenvalues-upper-half-plane-KdS} can be obtained by multiplying the angular ODE~\eqref{eq:angular-ode-KdS} by $\sin\theta\, \overline{\nu S_{m,\lambdabar}^{[s],\,\Xi,\nu}}$, integrating by parts and taking the real part:
\begin{align*}
&\frac{\Im (\lambdabar\,\overline\nu)}{-\Im\nu}\int_0^\pi \lp|S_{m,\lambdabar}^{[s],\,\Xi,\nu}(\theta)\rp|^2  \sin\theta  d\theta \\
&\quad=  \int_0^\pi 
\lp(\Delta_\theta\lp|\frac{dS_{m,\lambdabar}^{[s],\,\Xi,\nu}}{d\theta}\rp|^2+2(\Xi-1)\cos^2\theta(\mu+2s^2)\lp|S_{m,\lambdabar}^{[s],\,\Xi,\nu}(\theta)\rp|^2 \rp)\sin\theta  d\theta\\
&\quad\qquad+ \int_0^\pi \frac{1}{\Delta_\theta} \lp(
\frac{\lp(\Xi m+s\cos\theta\lp(\Xi-2(\Xi-1)\sin^2\theta\rp)\rp)^2}{\sin^2\theta}
+\Xi^3|\nu|^2\cos^2\theta\rp)\lp|S_{m,\lambdabar}^{[s],\,\Xi,\nu}(\theta)\rp|^2  \sin\theta  d\theta\,,
\end{align*}
where the right hand side is clearly non-negative if $\Im\nu>0$. In fact, it is only zero if $S_{m,\lambdabar}^{[s],\,\Xi,\nu}\equiv 0$. As trivial functions are not in the scope of the lemma, we obtain \eqref{eq:angular-eigenvalues-upper-half-plane-KdS}.
\end{proof}

\begin{remark} We note that, as before, $\lambdabar$ denotes a complex number with no restrictions whereas $\bm\uplambda$ denotes one of the eigenvalues identified in Lemma~\ref{lemma:angular-eigenvalues-KdS}.
\end{remark}


\subsubsection{The radial ODE and its boundary conditions}
\label{sec:radial-ODE-KdS}

Our starting point is the radial Teukolsky  ODE. In the literature, this ODE is  usually presented in terms of $\smlambdaXi{\upalpha}{s}:=\Delta^{-\frac{s+1}{2}}(r-r_3)\smlambdaXi{R}{s}$ as
\begin{equation}\label{eq:radial-ODE-alpha-KdS-actual-alpha}
\begin{split}
&\Delta^{-s} \frac{d}{dr}\lp(\Delta^{1+s}\frac{d}{dr}\rp)\smlambda{\upalpha}{s}+\frac{1}{\Delta}\lp[\Xi^2(\omega(r^2+a^2)-am)^2 - is \Xi (\omega(r^2+a^2)-am) \frac{d\Delta}{dr}\rp]\smlambdaXi{\upalpha}{s}\\
&\quad+\lp[4is\omega \Xi r-\frac{2}{L^2}(\mu+3s+2s^2)r^2+s\lp(1-\frac{a^2}{L^2}\rp)-\lambdabar+2\Xi^2 am\omega-a^2\Xi^2\omega^2\rp]\smlambdaXi{\upalpha}{s}=0\,,
\end{split}
\end{equation}
for $r\in(r_1,r_2)$. 
This equation has a regular singularity at  $r=\infty$ and singularities at each of the roots of $\Delta$, which are also regular if the Kerr-de Sitter parameters are subextremal. In this paper, we will consider the radial ODE in terms of the variable $\smlambdaXi{R}{s}$ introduced in \eqref{eq:separable-ansatz-KdS}, where \eqref{eq:radial-ODE-alpha-KdS-actual-alpha} becomes
\begin{align}
\begin{split}
&\Delta\frac{d^2}{dr^2}\smlambda{R}{s}+\frac{2\Delta}{r-r_3}\frac{d}{dr}\smlambda{R}{s}+\frac{1}{\Delta}\lp[\Xi(\omega(r^2+a^2)-am) - \frac12 is \frac{d\Delta}{dr}\rp]^2\smlambdaXi{R}{s} \\
&\qquad+\lp[\frac{M^2-\lp(1-\frac{a^2}{L^2}\rp)a^2}{\Delta}+4is\omega \Xi r-\frac{2(\mu-1+2s^2)}{L^2}r^2-\lambdabar-a^2\omega^2\Xi^2+2am\omega\Xi\rp]\smlambdaXi{R}{s}=0\,.
\end{split} \label{eq:radial-ODE-alpha-KdS}
\end{align}
As noted in \cite{Suzuki1998}, if $\mu=1$, equation \eqref{eq:radial-ODE-alpha-KdS} has singularities only at the roots of $\Delta$: $r=\infty$ is now a regular point of the ODE. Thus, from this point onwards, \textbf{we take $\mu\equiv 1$}.

Let us here introduce the notation
\begin{equation}\label{eq:def-ai}
\eta_j:=(-1)^{j}\frac{i\, (\omega-m\upomega_j)}{2\kappa_j}\,, \quad \upomega_j:= \frac{a}{r_j^2+a^2}\,,\quad \kappa_j:=\frac{1}{2 L^2\Xi  \left(r_j^2+a^2\right)}\left|\prod_{j'=0,j'\neq j}^3\left(r_j-r_{j'}\right)\right|\,,
\end{equation}
for $j\in\{0,1,2,3\}$. Note that we have $\sum_{j=0}^3 \eta_j =0$.  The quantities $ \upomega_j$ and $\kappa_j$ are, respectively, the angular velocity and the surface gravity of the horizon at $r=r_j$.

Based on the classical theory of regular singularities for ODEs, see \cite[Chapter 5]{Olver1973}, we can consider the following boundary conditions for \eqref{eq:radial-ODE-alpha-KdS}:
\begin{definition} \label{def:bdry-conditions-KdS} We say that  a solution, $\smlambdaXi{R}{s}(r)$, to \eqref{eq:radial-ODE-alpha-KdS} is 
\begin{itemize}
\item ingoing at $\mc{H}^+$ if the following are smooth at $r=r_1$:
\begin{itemize}
\item $\smlambdaXi{R}{ s}(r)(r-r_1)^{-\eta_1+\frac{s-1}{2}}$ if either $\Re\omega\neq m\upomega_1$ or $s\leq 0$, and
\item $\smlambdaXi{R}{ s}(r)(r-r_1)^{-\frac{s+1}{2}}$ if $\Re\omega=m\upomega_1$ and $s\geq 0$;
\end{itemize}
\item outgoing at $\mc{H}^+_c$ if the following are smooth at $r=r_2$:
\begin{itemize}
\item $\smlambdaXi{R}{ s}(r)(r-r_2)^{\eta_2-\frac{s+1}{2}}$ if either $\Re\omega\neq m\upomega_2$ or $s\geq 0$, and
\item $\smlambdaXi{R}{ s}(r)(r-r_2)^{\frac{s-1}{2}}$ if $\Re\omega=m\upomega_2$ and $s\leq 0$.
\end{itemize}
\end{itemize}
\end{definition}


\subsubsection{Precise definition of mode solution}

We are finally ready to define mode solution precisely:
\begin{definition} Fix $M>0$, $L>0$ and $|a|<L$ satisfying \eqref{eq:subextremal-condition-KdS}.  Take $s\in\frac12\mathbb{Z}$, $m-s\in\mathbb{Z}$ and $\omega\in\mathbb{C}\backslash\{0\}$ with $\Im \omega\geq 0$. Let $\swei{\upalpha}{s}(t,r,\theta,\phi)$ be a solution to \eqref{eq:Teukolsky-equation-KdS} which is given by \eqref{eq:separable-ansatz-KdS}. We say that $\swei{\upalpha}{s}(t,r,\theta,\phi)$ is a mode solution if
\begin{enumerate}
\item $\lambdabar=\bm\lambda_{ml\nu_0}^{[s],\,\Xi,\nu}$ for some $l\in\mathbb{Z}_{\geq \max\{|m|,|s|\}}$ and $\nu_0\in\mathbb{R}$, making $e^{im\phi}S_{m,\lambdabar}^{[s],\,\Xi,a\omega}$  a non-trivial smooth $s$-spin-weighted function on $\mathbb{S}^2$ such that $S_{m,\lambdabar}^{[s],\,\Xi,a\omega}$ is a normalized solution to the angular ODE \eqref{eq:angular-ode-KdS};
\item ${R}^{[s],\,\Xi, a, \omega}_{m,\lambdabar}(r)$ solves the radial ODE \eqref{eq:radial-ODE-alpha-KdS} with ingoing boundary conditions at $\mc{H}^+$ and outgoing boundary conditions at $\mc{H}^+_c$.
\end{enumerate}
\end{definition}

We note that Remark~\ref{rmk:why-mode-solution-def} still applies here, \textit{mutatis mutandis}.


\subsection{Some hidden spectral symmetries}
\label{sec:hidden-symmetries-KdS}

For $z_2>1$, consider the radial ODE
\begin{align} \label{eq:SQCD-ODE-KdS}
\begin{split}
&z(z-1)(z-z_2)\frac{d^2}{dz^2}y(z)+\frac{(4E-1)z_2+1}{4}y(z)\\
&\qquad -\frac{1}{(z-z_2)}\lp[\frac14z(z-1)((m_3+m_4)^2-1)-m_3m_4 (z-z_2)\lp(z-\frac12\rp)\rp]y(z)\\
&\qquad -\frac{1}{z(z-1)}\lp[m_1m_2\lp(\frac{z}{2}-z_2\rp)(z-1)+\frac14(z-z_2)((m_1+m_2)^2-1)\rp]y(z)=0\,,
\end{split}
\end{align}
where $m_1,m_2,m_3,m_4,E\in\mathbb{C}$. The radial ODE \eqref{eq:radial-ODE-alpha-KdS} may also be cast in this form:

\begin{lemma} \label{lemma:SQCD-rescaling-KdS} The radial ODE~\eqref{eq:radial-ODE-alpha-KdS} may be written as \eqref{eq:SQCD-ODE-KdS} if we let
\begin{align} \label{eq:mobius-transformation}
z:= z_\infty\frac{r-r_0}{r-{r}_3}
\,, \qquad z_\infty:=\frac{r_1-{r}_3}{r_1-r_0}\,,\qquad z_2:=z_\infty\frac{r_2-r_0}{r_2-{r}_3}\,,
\end{align}
and choose the parameters in \eqref{eq:SQCD-ODE-KdS} as follows: 
\begin{gather}\label{eq:SQCD-masses-KdS}
\begin{gathered}
m_1= s-\eta_1-\eta_0\,, \qquad m_2=\eta_0-\eta_1\,, \qquad m_3= -s-{\eta_1}-{\eta_0}\,, \qquad m_4 = \eta_0+\eta_1+2\eta_2\,, \\
\begin{split}
 Ez_2
%
%
&=\frac14(1+z_2)+(\eta_0+ \eta_1)^2z_2-(\eta_1^2+\eta_0\eta_1+\eta_1\eta_2-\eta_0\eta_2)+ \frac{L^2(\lambdabar-2\Xi^2 am\omega+a^2\Xi^2\omega^2)}{(r_2-r_3)(r_1-r_0)}\\
&\qquad+\frac{r_0^2 + r_1^2 + 2 r_0 r_2 + 2 r_2 (r_1 + r_2)}{2(r_2-r_3)(r_1-r_0)}s^2-\frac{r_1(r_2-r_3)+r_3(r_0+r_1)}{(r_1-r_0)(r_2-r_3)}\,,
\end{split}
\end{gathered}
\end{gather}
%
Furthermore, the boundary conditions in Definition~\ref{def:bdry-conditions-KdS} can be recast in terms of the new parameters, noting that 
\begin{gather*}
-\eta_1+\frac{s-1}{2}=\frac12(m_1+m_2-1)\,, \quad -\frac{s+1}{2}=-\frac12(m_1+m_2+1)\,, \\ \eta_2-\frac{s+1}{2}=\frac12(m_3+m_4-1)\,, \quad \frac{s-1}{2}=-\frac12(m_1+m_4+1)\,.
\end{gather*}

The choices of $m_1$ and $m_2$, and of $m_3$ and $m_4$ are not canonical: these boundary conditions and the ODE~\eqref{eq:SQCD-ODE-KdS} are invariant by the exchange of $m_1$ and $m_2$, and by the exchange of $m_3$ and $m_4$. 
\end{lemma}
\begin{proof} Under the M\"{o}bius transformation \eqref{eq:mobius-transformation}, the  singularities of \eqref{eq:radial-ODE-alpha-KdS-actual-alpha} transform as follows
\begin{gather*}
r=r_0\Rightarrow z=0\,, \qquad r=r_1\Rightarrow z=1\,,\qquad r={r}_3\Rightarrow z=\infty\,, \qquad
 r={r}_2\Rightarrow z=z_2 \,, \qquad r=\infty \Rightarrow z=z_\infty \,.
\end{gather*}
Note that $z_2<z_\infty$ as $r_3<r_0$, and furthermore  $z_2>1$, since
\begin{align*}
(z_2-1)(r_1-r_0)(r_2-r_3)
=(r_0-r_3)(r_2- r_1)> 0\,.
\end{align*}

The result then follows by direct computations, which have been independently obtained in the literature in \cite{Suzuki1998,Suzuki1999}, see also \cite{Hatsuda2021}.
\end{proof}

\begin{remark} To our knowledge, Lemma~\ref{lemma:SQCD-rescaling-KdS} has not appeared elsewhere in the literature. However, as in the case Lemma~\ref{lemma:SQCD-Kerr} in the Kerr case, the observation in Lemma~\ref{lemma:SQCD-rescaling-KdS} comes from comparing the Teukolsky equation in subextremal Kerr-de Sitter black holes and the $SU(2)$ Seiberg--Witten theory, but now  with four fundamental hypermultiplets, in supersymmetric quantum chromodynamics \cite{Ito2017}. We carry out this comparison in analogy to those drawn throughout Aminov, Grassi and Hatsuda's paper \cite{Aminov2020}, and present here the final outcome as Lemma~\ref{lemma:SQCD-rescaling-KdS}.
\end{remark}

\begin{remark}
Taking the $L\to \infty$ limit of Lemma~\ref{lemma:SQCD-rescaling-KdS} yields Lemma~\ref{lemma:SQCD-Kerr}. For further details about the confluence process which turns Heun equations into confluent Heun equations, we refer the reader to the classical books \cite{Ronveaux1995,Slavyanov2000} and references therein.
\end{remark}

Next, we give a characterization of the point spectrum of \eqref{eq:SQCD-ODE-KdS}, in the space of solutions with suitable boundary conditions. Though other methods such as Jaffé-type expansions are sometimes considered in the literature, see \cite{Yoshida2010}, we find it convenient to replace the Jaffé polynomials of Proposition~\ref{prop:Jaffe-expansion} with hypergeometric polynomials. 

\begin{proposition}[Point spectrum] \label{prop:hypergeometric-expansion} 
Let $E\in\mathbb{C}$. Consider a quadruple $(m_{i_1},m_{i_2},m_{i_3},m_{i_4})$, where $i_1$, $i_2$, $i_3$ and $i_4$ are distinct natural numbers between 1 and 4, of complex numbers verifying three conditions: $\sum_j m_j\neq N$ for $N\in\mathbb{Z}_{\geq 2}$, $\Re(m_{i_1}+m_{i_2}-1)<0$ if $\Im(m_{i_1}+m_{i_2})=0$, and $\Re(m_{i_3}+m_{i_4}-1)<0$ if $\Im(m_{i_3}+m_{i_4})=0$. 

Let ${R}_{(m_{i_1},m_{i_2},m_{i_3},m_{i_4})}^{E}$ be the unique, up to rescaling, solution of the ODE, c.f.\ \eqref{eq:SQCD-ODE-KdS},
\begin{align*}
&z(z-1)(z-z_2)\frac{d^2}{dz^2}{R}(z)+\frac{(4E-1)z_2+1}{4}{R}(z)\\
&\qquad -\frac{1}{(z-z_2)}\lp[\frac14z(z-1)((m_{i_3}+m_{i_4})^2-1)-m_{i_3}m_{i_4} (z-z_2)\lp(z-\frac12\rp)\rp]{R}(z)\\
&\qquad -\frac{1}{z(z-1)}\lp[m_{i_1}m_{i_2}\lp(\frac{z}{2}-z_2\rp)(z-1)+\frac14(z-z_2)((m_{i_1}+m_{i_2})^2-1)\rp]{R}(z)=0\,,
\end{align*}
with boundary conditions
\begin{itemize}
\item ${R}(z)(z-1)^{\frac12(m_{i_1}+m_{i_2}-1)}$ is smooth at $z=1$,
\item ${R}(z)(z-z_2)^{\frac12(m_{i_3}+m_{i_4}-1)}$ is smooth  at $z=z_2$.
\end{itemize}
Then ${R}_{(m_{i_1},m_{i_2},m_{i_3},m_{i_4})}^{E}$ is nontrivial if and only if the continued fraction condition
\begin{align}
0= A_0^{(0)}+\frac{-A_0^{(+1)}A_1^{(-1)}}{ A_1^{(0)}+\frac{-A_1^{(+1)}A_2^{(-1)}}{ A_2^{(0)}+\frac{-A_2^{(+1)}A_3^{(-1)}}{\cdots} }} \label{eq:hypergeometric-expansion-continued-fraction}
\end{align}
holds, where the coefficients $A^{(0)}_n$ and $A_n^{(\pm 1)}$ satisfy
\begin{align}
\begin{split}
A_{n-1}^{(+1)} A_{n}^{(-1)} &=\frac{n^2(n-\sigma_1(\bm{m}))^2 \lp[\sigma_1(\bm{m}) \sigma_3(\bm{m}) - 4 \sigma_4(\bm{m}) + \lp(\sigma_2(\bm{m}) +  (n-\sigma_1(\bm{m}))n\rp)^2\rp]}{(\sigma_1(\bm{m})-2n)^2[(\sigma_1(\bm{m})-2n)^2-1]} \\
&\qquad + \frac{n(n-\sigma_1(\bm{m}))\lp[\sigma_1(\bm{m}) \lp(\sigma_2(\bm{m}) \sigma_3(\bm{m}) - \sigma_1(\bm{m}) \sigma_4(\bm{m})\rp)-\sigma_3(\bm{m})^2\rp]}{(\sigma_1(\bm{m})-2n)^2[(\sigma_1(\bm{m})-2n)^2-1]}\,,\\
A_n^{(0)}&= \frac{[\sigma_1(\bm{m})^3 - 4 \sigma_1(\bm{m})\sigma_2(\bm{m}) + 8 \sigma_3(\bm{m})]\sigma_1(\bm{m})}{8 (-2 + \sigma_1(\bm{m}) - 2 n) (\sigma_1(\bm{m}) - 2 n)}+\frac{(4E+1)z_2+1}{4(z_2-1)}\\
&\qquad+\frac{(z_2+1)(\sigma_1(\bm{m})-2n-2) (\sigma_1(\bm{m})-2n)}{8(z_2-1)}-\frac{z_2[\sigma_1(\bm{m})^2-2\sigma_2(\bm{m})]}{4(z_2-1)}\,,
\end{split} \label{eq:hypergeometric-expansion-recursion-coefficients}
\end{align}
and where $\sigma_{1}(\bm{m}):= m_1+m_2+m_3+m_4$, $\sigma_{2}(\bm{m}):=m_1 m_2+m_1 m_3+m_1m_4+m_2m_3+m_2m_4+m_3m_4$, $\sigma_{3}(\bm{m}):=m_1 m_2 m_3+m_1m_2m_4+m_2m_3m_4$, and $\sigma_{3}(\bm{m}):=m_1 m_2 m_3m_4$ are symmetric polynomials in $\bm{m}=(m_{i_1},m_{i_2},m_{i_3},m_{i_4})$.
\end{proposition}

\begin{proof} 
From the classical theory of ODEs \cite[Chapter 5]{Olver1973}, the first condition on  ${R}_{(m_{i_1},m_{i_2},m_{i_3},m_{i_4})}^{E}(z)$ is fulfilled only if ${R}_{(m_{i_1},m_{i_2},m_{i_3},m_{i_4})}^{E}(z)(z-1)^{\frac12(m_{i_1}+m_{i_2}-1)}$ is holomorphic in $|z-1|<\min\{1,z_2-1\}$ whereas the second one holds if and only if  ${R}_{(m_{i_1},m_{i_2},m_{i_3},m_{i_4})}^{E}(z)(z-z_2)^{\frac12(m_{i_3}+m_{i_4}-1)}$ is holomorphic for $|z-z_2|<z_2-1$. Hence, since these regions have nontrivial intersection, we conclude that ${R}_{(m_{i_1},m_{i_2},m_{i_3},m_{i_4})}^{E}(z)(z-1)^{\frac12(m_{i_1}+m_{i_2}-1)}(z-z_2)^{\frac12(m_{i_3}+m_{i_4}-1)}$ must be holomorphic in an ellipse 
\begin{align*}
\mathscr E:=\lp\{\begin{array}{lr}
\displaystyle\frac{4\lp(\Re z-\frac{z_2+1}{2}\rp)^2}{(z_2+1)^2}+\frac{(8z_2(z_2-1)-4)(\Im z)^2}{(1+z_2)^2(2z_2-3)}< 1\,,& \text{~if~} z_2\geq 2\\
\displaystyle \frac{4\lp(\Re z-\frac{z_2+1}{2}\rp)^2}{9(z_2-1)^2}+\frac{4(\Im z)^2}{3(z_2-1)^2}< 1\,, & \text{~if~} 1<z_2< 2
\end{array}\rp\}\,.
\end{align*}
Thus,  $g$ defined through
\begin{equation*}
{g}(z):= z^{\pm \frac12(m_{i_1}-m_{i_2}\mp 1)}(z-1)^{\frac{1}{2}(m_{i_1}+m_{i_2}-1)}(z-z_2)^{\frac12(m_{i_3}+m_{i_4}-1)}(z-z_\infty)^{\frac{3}{2}(s+1)}{R}_{(m_{i_1},m_{i_2},m_{i_3},m_{i_4})}^E(z)\,,
\end{equation*}
is the unique (up to rescaling) solution of the ODE
\begin{align}
\begin{gathered}
\mathscr L g=0\,,\\ \mathscr{L}:= z(z-1)(z-z_2)\frac{d^2}{dz^2}+\lp[\gamma(z-1)(z-z_2)+\delta z(z-z_2)+\epsilon z(z-1)\rp]\frac{d}{dz}+\alpha \beta z -q\,, 
\end{gathered} \label{eq:Heun-ODE}
\end{align}
which is holomorphic in the ellipse $\mathscr E$. Here, the  parameters in $\mc L$ satisfy
\begin{gather*}
\gamma = 1\pm (m_{i_1}-m_{i_2})\,, \quad \delta = 1-m_{i_1}-m_{i_2}\,, \quad \epsilon=1-m_{i_3}-m_{i_4}\,, \qquad \omega_H=1-{\textstyle\sum_{j=1}^4}m_j\,,\\
\alpha=  1-m_{i_1}\frac{1\mp 1}{2}-m_{i_2}\frac{1\pm 1}{2}-m_{i_3}\,, \quad \beta=1-m_{i_1}\frac{1\mp 1}{2}-m_{i_2}\frac{1\pm 1}{2}-m_{i_4}\,,\numberthis \label{eq:Heun-parameters-KdS}\\
\gamma(\delta-\epsilon)+\epsilon(1-\omega_H)+2\alpha\beta=2-{\textstyle \sum_{j=1}^4} m_j^2+2m_{i_2}(m_{i_2}-1)(1\pm 1)+2m_{i_1}(m_{i_1}-1)(1\mp 1)\,,\\
4q-4\alpha\beta+\omega^2_H = -1+{\textstyle \sum_{j=1}^4} m_j^2+\lp[2m_{i_2}(m_{i_2}-1)(1\pm 1)+2m_{i_1}(m_{i_1}-1)(1\mp 1)+1\rp](z_2-1)-4Ez_2\,,
\end{gather*}
where $\epsilon+\delta-1=\alpha+\beta-\gamma$ and we have used the notation $\omega_H:=\delta+\epsilon-1=\alpha+\beta-\gamma$.

\medskip
\noindent \textit{Step 1: a formal expansion.} As in Proposition~\ref{lemma:hidden-symmetry}, it is convenient  to identify a set of adequate special functions in which to expand solutions to \eqref{eq:Heun-ODE}. Our construction is loosely based on a small modification of \cite[Part A, Chapter 4]{Ronveaux1995} (see also \cite{Svartholm1939} and \cite{Erdelyi1944}). Let us introduce the operators
\begin{align*}
\Lambda_1 &:= (z-1)(z-z_2)\frac{d^2}{dz^2} +\lp[\delta(z-z_2)+\epsilon (z-1)\rp]\frac{d}{dz} = t(t-1)\frac{d^2}{dt^2}+\lp[\delta(t-1)+\epsilon\, t\rp]\frac{d}{dt}\,,\\
\tilde \Lambda_2 &:=(z-1)(z-z_2)\frac{d}{dz} = (z_2-1)\Lambda_2\,, \qquad \Lambda_2:= t(t-1) \frac{d}{dt}\,,
\end{align*}
where $t:=\frac{z-1}{z_2-1}$. Then, $\mathscr{L}$ may be written as 
\begin{align}
\mathscr{L} &= z\Lambda_1+\gamma \tilde{\Lambda}_2 +\alpha \beta z -q= (z_2-1)\lp[t\Lambda_1 + \gamma\Lambda_2 +\alpha\beta t-\frac{q-\alpha\beta-\Lambda_1}{z_2-1}\rp]\,. \label{eq:Heun-in-terms-of-Lambdas}
\end{align}
Let us introduce the hypergeometric functions
\begin{align*}
F_\nu (z) := {_{2}F_{1}}\lp(-\nu+\frac{\omega_H}{2}-\frac{1}{2}, \nu+\frac{\omega_H}{2}+\frac{1}{2}; \delta; \frac{z-1}{z_2-1}\rp)={_{2}F_{1}}\lp(-\nu+\frac{\omega_H}{2}-\frac{1}{2}, \nu+\frac{\omega_H}{2}+\frac{1}{2}; \delta; t\rp)\,,
\end{align*}
which satisfy
\begin{align*}
\Lambda_1 F_\nu &= \lp(\nu+\frac12 -\frac{\omega_H}{2}\rp) \lp(\nu+\frac12+\frac{\omega_H}{2}\rp) F_\nu\,,\numberthis\label{eq:Lambda-1-formally}\\
\frac{\mathscr L F_\nu}{z_2-1}&= \lp(\nu+\frac12 -\frac{\omega_H}{2}\rp) \lp(\nu+\frac12+\frac{\omega_H}{2}\rp) tF_\nu + \gamma \Lambda_2 F_\nu +\alpha\beta t-\frac{q-\alpha\beta-\lp(\nu+\frac12 -\frac{\omega_H}{2}\rp) \lp(\nu+\frac12+\frac{\omega_H}{2}\rp)}{z_2-1}F_\nu \\
&= \tilde A_\nu^{(+1)} F_{\nu+1}+\tilde A_\nu^{(0)}F_\nu +\tilde A_\nu^{(-1)} F_{\nu -1}
\,,
\end{align*}
where the coefficients $\tilde A_\nu^{(\pm 1)}$ and $\tilde A_\nu^{(0)}$ are given by
\begin{align*}
\tilde A_\nu^{(+1)} &= -\frac{\lp(\nu+\frac12 +\frac{\omega_H}{2}\rp)\lp(\nu+\frac12 -\frac{\omega_H}{2}+\alpha\rp)\lp(\nu+\frac12 -\frac{\omega_H}{2}+\beta\rp)\lp(\nu+\frac12 -\frac{\omega_H}{2}+\delta\rp)}{2(2\nu+1)(\nu+1)}\,,\\
\tilde A_\nu^{(0)} 
&= \frac{(1-\omega_H)(\delta-\epsilon)[(\alpha-\beta)^2-(\gamma-1)^2]}{32\nu(\nu+1)}+\frac{1}{2}\lp(1 +\frac{2}{z_2-1}\rp)\nu(\nu+1)\\
&\qquad+\frac{1}{4}[\gamma(\delta-\epsilon)+\epsilon(1-\omega_H)+2\alpha\beta]-\frac{q-\alpha\beta}{z_2-1}-\frac{1}{4(z_2-1)}(\omega_H^2-1)\,,\numberthis \label{eq:general-recursion-coefficients-KdS}\\
\tilde A_\nu^{(-1)} &= -\frac{\lp(\nu+\frac12 -\frac{\omega_H}{2}\rp)\lp(\nu+\frac12 +\frac{\omega_H}{2}-\alpha\rp)\lp(\nu+\frac12 +\frac{\omega_H}{2}-\beta\rp)\lp(\nu+\frac12 +\frac{\omega_H}{2}-\delta\rp)}{2\nu(2\nu+1)}\,,\\
\tilde A_{\nu-1}^{(+1)}\tilde A_{\nu}^{(-1)}&=\frac{\lp[\nu^2-\frac{(\omega_H-1)^2}{4}\rp]\lp[\nu^2-\lp(\frac{\omega_H+1}{2}-\alpha\rp)^2\rp]\lp[\nu^2-\lp(\frac{\omega_H+1}{2}-\beta\rp)^2\rp]\lp[\nu^2-\lp(\frac{\omega_H+1}{2}-\delta\rp)^2\rp]}{4(2\nu+1)\nu^2(2\nu-1)}\,.
\end{align*}
Now, suppose that the function
\begin{gather*}   
y(z) = \sum_{n=-\infty}^\infty b_n F_{N}(z)\,, \qquad N:= \nu+n\,, 
\end{gather*}
is well defined and its derivatives are obtained by differentiating term by term. If $y$ solves $\mathscr L y=0$, it follows that $b_n$ verify the recursive relation
\begin{gather}   
A_n^{(+1)} b_{n+1} + A_n^{(0)} b_n +A_n^{(-1)} b_{n-1}=0\,,\label{eq:general-recursion-KdS}
\end{gather}
where we have $A_n^{(+1)}:=\tilde A_{N+1}^{(-1)}$, $A_n^{(+1)}:=\tilde A_{N-1}^{(+1)}$  and $A_n^{0}=\tilde A_N^{(0)}$. These coefficients can be readily computed from \eqref{eq:Heun-parameters-KdS} and \eqref{eq:general-recursion-coefficients-KdS}. 

Up to this point, $\nu$ has remained a free parameter; let us now fix $\nu$ to be $\nu=(\omega_H-1)/2$. One can then check that in \eqref{eq:general-recursion-KdS}, we shall have $b_n=0$ for $n< 0$. Therefore, our final ansatz for a solution of the ODE \eqref{eq:Heun-ODE} is
\begin{align}
g(z) := \sum_{n=0}^\infty b_n f_{n}(z)\,, \label{eq:general-expansion-KdS}
\end{align} 
where $b_n$ satisfy \eqref{eq:general-recursion-KdS} and $f_{n}(z)= F_{(\omega_H-1)/2+n}(z)$ are polynomials of degree $n$ in $z-1$: denoting by $(\cdot)_k$ the rising factorial,
\begin{align}
f_n(z):={_{2}F_{1}}\lp(-n, n+\omega_H; \delta; \frac{z-1}{z_2-1}\rp)=\frac{(\delta)_n}{n!}P_n^{(\delta-1,\epsilon-1)}\lp(1-\frac{2(z-1)}{z_2-1}\rp)\,, \label{eq:basis-hypergeometric-polynomials}
\end{align}
where $P_n^{(\cdot,\cdot)}$ denote the usual Jacobi polynomials.

To compute the region of convergence of the series \eqref{eq:general-expansion-KdS}, we need to examine the asymptotics of both $f_{n}(z)$ and $b_n$ as $n\to \infty$. We begin with the latter: a large $n$ asymptotic analysis of the recursion relation following  \cite[Theorem 2.3(b)]{Gautschi1967} shows that, as long as $b_n\not\equiv 0$,  we have either 
\begin{align*}
\lim_{n\to \infty}\frac{b_{n+1}}{b_n}=-\frac{1+\sqrt{z_2}}{1-\sqrt{z_2}}=1+\frac{2}{z_2-1}\lp(\sqrt{z_2}-1\rp)>1\,,
\end{align*}
or
\begin{align}
\lim_{n\to \infty}\frac{b_{n+1}}{b_n}=-\frac{1-\sqrt{z_2}}{1+\sqrt{z_2}}=1-\frac{2}{z_2-1}\lp(\sqrt{z_2}+1\rp) <1\,. \label{eq:minimal-solution-KdS}
\end{align}
From the above, it is clear that \eqref{eq:minimal-solution-KdS} defines the so-called minimal solution to the recursive relation \eqref{eq:general-expansion-KdS}. 

Let us now turn to the asymptotics of the functions $f_n(z)$ for large $n$: from \cite[Equation (4.7)]{Erdelyi1944} and \cite[Equation (4.3.3)]{Ronveaux1995}, we obtain that
\begin{align*}
 \quad \lim_{n\to\infty} \lp|\frac{f_{n+1}}{f_n}\rp|=
 \begin{dcases}\lp|\frac{1+Z(z)}{1-Z(z)}\rp|\,, \quad Z(z):=\lp(\frac{z-z_2}{z-1}\rp)^{1/2}\,, \quad &z\notin (1,z_2)\,,\\
 1\,, &z\in (1,z_2)
 \end{dcases} 
 \,, 
\end{align*}
where the last line can be obtained by a limiting procedure. We deduce, in agreement with \cite[page 417]{Svartholm1939}, that \eqref{eq:general-expansion-KdS} converges in the interval $[1,z_2]$ if and only if  \eqref{eq:minimal-solution-KdS} holds, and so do the series obtained by differentiating \eqref{eq:general-expansion-KdS} term by term up to two times (see, for instance, \cite[Part A, equation (4.2.16)]{Ronveaux1995}). In fact, in this case, \eqref{eq:general-expansion-KdS} converges and yields a holomorphic solution in the entire  ellipse
\begin{align*}
\mathscr E':=\lp\{\frac{\lp(\Re z-\frac{z_2+1}{2}\rp)^2}{(\frac{z_2+1}{2})^2}+\frac{(\Im z)^2}{z_2}< 1\rp\}\,,
\end{align*}
which contains $\mathscr E$, see Figure~\ref{fig:ellipses}. By construction, the first and second derivatives of the holomorphic solution are obtained by term-by-term differentiation.

\begin{figure}[htbp]
\centering
\includegraphics[scale=.6]{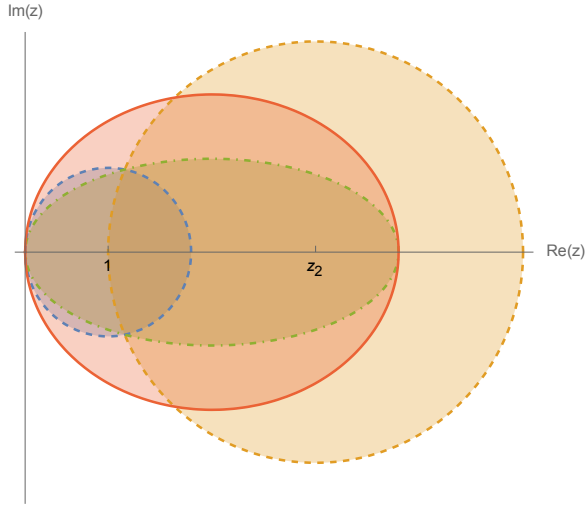}
\caption{Ellipses $\mathscr E$ (green region in the interior of dash-dotted line) and $\mathscr E'$ (red  region in the interior of full line)  in proof of Proposition~\ref{prop:hypergeometric-expansion}. The circles with dashed boundary in blue and orange represent the regions of analyticity of local power series solutions of $\mathscr L g=0$ around $z=1$ and $z=z_2$, respectively.} 
\label{fig:ellipses}
\end{figure}

%

\medskip
\noindent \textit{Step 2: from \eqref{eq:hypergeometric-expansion-continued-fraction} to $g\not\equiv 0$.} If the continued fraction equation  \eqref{eq:hypergeometric-expansion-continued-fraction} is satisfied, by classical results in the theory of three-term recurrence relations (see \cite[Theorem 1.1]{Gautschi1967}), then there exists a nonzero  minimal solution to the recursion \eqref{eq:general-recursion-KdS}. Thus, there are coefficients $b_n\not\equiv 0$ such that \eqref{eq:minimal-solution-KdS} holds, implying that, through \eqref{eq:general-expansion-KdS}, we may construct a solution, $g\not\equiv 0$, to the ODE \eqref{eq:Heun-ODE} which is holomorphic in $\mathscr E$.

\medskip

\noindent \textit{Step 3: from  $g\not\equiv 0$ to  \eqref{eq:hypergeometric-expansion-continued-fraction}.} It is a classical result, see \cite[Theorem 9.1.1]{Szego1939} and the more recent \cite{Carlson1974}, that a function which is analytic in an interval admits an expansion in Jacobi polynomials (in addition to, of course, power series), with large freedom in the choice of Jacobi polynomial parameters, which converges in an ellipse around the interval of analyticity. It is convenient to introduce a coordinate adapted to the Jacobi polynomials,
\begin{align*}
x:=1-\frac{2(z-1)}{z_2-1}\,;
\end{align*}
we will always assume $x=x(z)$ in what follows. Under the conditions on $m_i$ that we impose, the theory guarantees that 
\begin{align*}
g(z)&=\sum_{n=0}^\infty \frac{2^n}{(n+\omega_H)_n}G^{(n)}(n+\delta,n+\epsilon;-1,+1)P_n^{(\delta-1,\epsilon-1)}(x) \,, 
\end{align*}
converges and defines a holomorphic function in the ellipse $\mathscr E$, yielding a representation of the form \eqref{eq:general-expansion-KdS} for $g$ in $\mathscr E$: by \eqref{eq:basis-hypergeometric-polynomials}
\begin{equation}
g(z)=\sum_{n=0}^\infty b_n f_n(z)\,, \qquad b_n\equiv \frac{2^nn!}{(\delta)_n(n+\omega_H)_n}G^{(n)}(n+\delta,n+\epsilon;-1,+1)\not\equiv 0\,. \label{eq:bn-cn}
\end{equation}
In the previous formulas, the coefficients $G^{(n)}$ are given by the integral
\begin{align*}
G^{(n)}(n+\delta,n+\epsilon;-1,+1):=\frac{n!}{2\pi i}\int_\Gamma g(z) R_{-n-1}(n+\delta,n+\epsilon;x(z)+1,x(z)-1)dz\,,
\end{align*}
where $\Gamma$ is a closed contour in $\mathscr E$ which  contains the interval $z\in[1,z_2]\Leftrightarrow x\in[-1,1]$,  
\begin{align}
R_N(B_1,B_2;y_1,y_2)=\int_0^1[u y_1+(1-u)y_2]^Nu^{B_1-1}\frac{(1-u)^{B_2-1}}{\mathbb{B}(B_1,B_2)} du\,, \label{eq:R-function}
\end{align}
and $\mathbb{B}$ denotes the beta function. 

It is worth noting that the functions in \eqref{eq:R-function} are, unsurprisingly, closely related to Jacobi polynomials and verify many of the same identities. Such identities can be used to show that the term-by-term manipulations in step 1 of the series representing $g$ hold. As an example, in view of the decomposition \eqref{eq:Heun-in-terms-of-Lambdas}, let us try to find a representation for $\Lambda_1g$. By assumption, this is a holomorphic function in the ellipse $\mathscr E$ and thus, by the aforementioned classical theory, it can be written as
\begin{align*}
\Lambda_1 g(z)&=\sum_{n=0}^\infty \frac{2^n}{(n+\omega_H)_n}\tilde{G}^{(n)}(n+\delta,n+\epsilon;-1,+1)P_n^{(\delta-1,\epsilon-1)}(x) \,,
\end{align*}
with  coefficients given by  
\begin{align*}
\tilde{G}^{(n)}(n+\delta,n+\epsilon;-1,+1)&:=\frac{n!}{2\pi i}\int_\Gamma \Lambda_1g (z) R_{-n-1}(n+\delta,n+\epsilon;x(z)+1,x(z)-1)dz\\
&=-\frac{n!}{2\pi i}\int_\Gamma g(z) \Lambda_3 R_{-n-1}(n+\delta,n+\epsilon;x(z)+1,x(z)-1)dz\,.
\end{align*}
For the last equality, we have integrated by parts,  keeping in mind that $\Gamma$ lies in the region of holomorphicity of $\Lambda_1g$ and is a closed contour, and we have used the notation
\begin{align*}
\Lambda_3 := \frac{d}{dx}\lp[(1-x^2)\frac{d}{dx}\rp]+\lp(-(1-\delta)(1+x)+(1-\epsilon)(1-x)\rp)\frac{d}{dx}-(1-\omega_H)\,.
\end{align*}
From the identity
\begin{align*}
\Lambda_3 R_{-n-1}(n+\delta,n+\epsilon;x+1,x-1)= -n(n+\omega_H)R_{-n-1}(n+\delta,n+\epsilon;x+1,x-1)\,, 
\end{align*}
which follows from \eqref{eq:R-function}, we now deduce that 
\begin{align*}
g(z)=\sum_{n=0}^\infty b_n f_n(z)\implies \Lambda_1 g(z)=\sum_{n=0}^\infty n(n+\omega_H)b_n f_n(z)\,,
\end{align*}
consistently with \eqref{eq:Lambda-1-formally} after taking $\nu\to n+(\omega_H-1)/2$. 

By arguing similarly for the remaining terms of $\mathscr L$ in \eqref{eq:Heun-in-terms-of-Lambdas}, we deduce that the coefficients $b_n$ must satisfy the recursion relation \eqref{eq:general-recursion-KdS}. Then, by the convergence of the series in \eqref{eq:bn-cn} and the convergence analysis at the end of step 1, $b_n$ also need to satisfy the asymptotic relation \eqref{eq:minimal-solution-KdS}, i.e.\ they must be a minimal solution  to the recursion \eqref{eq:general-recursion-KdS}. Once again, the classical result \cite[Theorem 1.1]{Gautschi1967} then implies that the associated continued fraction equation  \eqref{eq:hypergeometric-expansion-continued-fraction} (see Proposition~\ref{lemma:hidden-symmetry} for details on this relationship) must be satisfied.
\end{proof}

From the point of view of the Teukolsky equation, Proposition~\ref{prop:hypergeometric-expansion} is the statement that its point spectrum, in the space of solutions with appropriate boundary conditions, is invariant under exchanges $m_i\leftrightarrow m_j$. The precise statement we will use is:

\begin{corollary}[Hidden spectral symmetries] \label{lemma:hidden-symmetry-KdS}
Fix $M>0$, $L>0$ and $|a|<L$ satisfying \eqref{eq:subextremal-condition-KdS}. Let $s\in\frac12\mathbb{Z}$, $m-s\in\mathbb{Z}$, $\lambdabar\in\mathbb{C}$ and $\omega\in\mathbb{C}$ satisfy $\Im\omega\geq 0$. If $\Re\omega\neq m\upomega_1,m\upomega_2$, the following are equivalent:
\begin{enumerate}[label=(\roman*)]
\item  (original case and $m_1\leftrightarrow m_2$ symmetry) there is a nontrivial solution to the radial ODE~\eqref{eq:radial-ODE-alpha-KdS} which is ingoing at $\mc{H}^+$  and outgoing at $\mc{H}^+_c$, i.e.\ which satisfies the boundary conditions
\begin{align*}
&\smlambdaXi{R}{s}(z)(z-1)^{-\eta_1+\frac{s-1}{2}} \text{~smooth at $z=1$}\,, \\ 
&\smlambdaXi{R}{s}(z)(z-z_2)^{\eta_2-\frac{s+1}{2}} \text{~smooth at $z=z_2$}\,;
\end{align*}
\item \label{it:case-m1-to-m3-KdS}($m_1\leftrightarrow m_3$ symmetry)  there is a nontrivial solution to the radial ODE equal to~\eqref{eq:radial-ODE-alpha-KdS} but where $s$ is replaced by $-s$, which satisfies the boundary conditions
\begin{align*}
&\smlambdaXi{ R}{s}(z)(z-1)^{-\eta_1-\frac{s+1}{2}}\text{~smooth at $z=1$}\,, \\ 
&\smlambdaXi{ R}{s}(z)(z-z_2)^{\eta_2+\frac{s-1}{2}} \text{~smooth at $z=z_2$}\,;
\end{align*}
\item  \label{it:case-m2-to-m3-KdS}($m_2\leftrightarrow m_3$ symmetry) there is a nontrivial solution to the radial ODE 
\begin{equation}\label{eq:tilded-ODE-R-KdS}
\begin{split}
&z(z-1)(z-z_2)\frac{d^2}{dz^2} \smlambdaXi{\tilde R}{s}(z)
+\frac{L^2(\lambdabar-2a\Xi^2m\omega+a^2\Xi^2\omega^2)}{(r_1-r_0)(r_2-r_3)} \smlambdaXi{\tilde R}{s}(z)\\
&\quad+\frac{1}{(r_1-r_0)(r_2-r_3)z}\lp[s^2(r_1+r_2)^2\lp(z-\frac{(r_2-r_0)(r_1-r_3)}{r_1^2+2r_2 r_1+r_2^2}\rp)+\frac12 z(r_0+r_1)^2\rp]\smlambda{\tilde R}{s}(z)\\
&\quad+\frac{z}{(z_2-z)}\lp[(z-1)\eta_2^2+2(z_2-1)\eta_0\eta_2 +\frac{(z_2-1)^2}{(z-1)}\eta_0^2-\frac{1}{4}(z-1)\rp] \smlambdaXi{\tilde R}{s}(z)\\
&\quad+\frac{z}{(z-1)}\lp[2(z_2-1)\eta_0\eta_1+(z_2-z)\eta_1^2-\frac{(z_2-z)}{4z^2}-2(z-1)\eta_1\eta_2\rp] \smlambdaXi{\tilde R}{s}(z)=0\,.
\end{split}
\end{equation}
satisfying the boundary conditions
\begin{equation}\label{eq:case-m2-to-m3-bdry-conditions-KdS}
\begin{split}
&\smlambdaXi{\tilde R}{s}(z)(z-1)^{-\eta_0-\eta_1-\frac{1}{2}} \text{~smooth at $z=1$}\,, \\ 
&\smlambdaXi{\tilde R}{s}(z)(z-z_2)^{\eta_0+\eta_2-\frac{1}{2}} \text{~smooth at $z=z_2$}\,;
\end{split}
\end{equation}
\end{enumerate}
If $\Re\omega=m\upomega_1$, then (i) and (iii) are equivalent if $s\leq 0$ and, if $s\geq 0$, (ii) and (iii) are equivalent. In turn, if $\Re\omega=m\upomega_2$, then (i) and (iii) are equivalent if $s\geq 0$ and, if $s\leq 0$, (ii) and (iii) are equivalent.
\end{corollary}
\begin{proof} Let us note that, since
$$\sum_{i}m_i=2(\eta_2-\eta_1)=i\lp(\frac{\Re\omega-m\upomega_1}{\kappa_1}+\frac{\Re\omega-m\upomega_2}{\kappa_2}\rp)-\Im\omega\lp(\frac{1}{\kappa_1}+\frac{1}{\kappa_2}\rp)\,,$$
the conditions of Proposition~\ref{prop:hypergeometric-expansion} clearly hold under our assumptions.
Thus, the conclusion follows from Proposition~\ref{prop:hypergeometric-expansion} after applying Lemma~\ref{lemma:SQCD-rescaling-KdS} and rewriting the ODE \eqref{eq:tilded-ODE-R-KdS} and the boundary conditions in terms of $(m_1,m_2,m_3,m_4,E)$. 
\end{proof}

We emphasize that Proposition~\ref{prop:hypergeometric-expansion} has not, to our knowledge, appeared elsewhere in the literature even as a conjecture. Remark \ref{rmk:spin-reversal-symmetry} for the Kerr case also applies in the Kerr-de Sitter case considered above, \textit{mutatis mutandis}.


\subsection{A partial mode stability result}
\label{sec:proof-KdS}

In this section, we give a proof of Theorem~\ref{thm:partial-mode-stability-KdS-intro}. To be precise, we prove
\begin{theorem}\label{thm:partial-mode-stability-KdS} Fix $M>0$, $|a|< 3/\Lambda$  and $\Lambda>0$ satisfying \eqref{eq:subextremal-condition-KdS}, $\frac12s\in\mathbb{Z}$, $m-s\in\mathbb{Z}$ and $(\omega,\lambdabar)$ such that one of the following holds:
\begin{itemize}
\item $\Im\omega>0$, $\Im(\lambdabar\, \overline{\omega})\leq 0$  and  $|\omega|\not\in|m| \lp(0,\frac{\upomega_2/\kappa_2+\upomega_0/\kappa_0}{1/\kappa_2+1/{\kappa_0}}\rp)\subsetneq |m|(0,\upomega_1)$;
\item $\omega\in\mathbb{R}$, $\lambdabar\in\mathbb{R}$ and, if $|s|\neq \frac12,\frac32$, additionally $m=0$ or $\frac{\omega}{m}\not\in \lp(\frac{\upomega_1/\kappa_1-\upomega_0/\kappa_0}{1/{\kappa_1}-1/{\kappa_0}},\frac{\upomega_2/\kappa_2+\upomega_0/\kappa_0}{1/\kappa_2+1/{\kappa_0}}\rp)\subsetneq (\upomega_2,\upomega_1)$.
\end{itemize}
If $\smlambda{\upalpha}{s}(r)$ is a solution to \eqref{eq:radial-ODE-alpha-KdS} with respect to these parameters which is ingoing at $\mc H^+$ and outgoing at $\mc H^+_c$, then $\smlambda{\upalpha}{s}\equiv 0$.
\end{theorem}

\begin{remark} Note that, by Lemma~\ref{lemma:angular-eigenvalues-KdS}, it is clear that the conditions on $\lambdabar$ which we impose here are verified by the angular eigenvalues associated to mode solutions. Hence, Theorem~\ref{thm:partial-mode-stability-KdS} implies, and is stronger than, Theorem~\ref{thm:partial-mode-stability-KdS-intro}.
\end{remark}

\begin{proof}[Proof of Theorem~\ref{thm:partial-mode-stability-KdS}]  The proof follows two steps. First, we analyze the original radial ODE~\eqref{eq:radial-ODE-alpha-KdS} and identify superradiant frequencies, for which the solutions considered in the statement may exist. Then, we attempt to rule out some of these solutions by making use of the hidden symmetries uncovered by Proposition~\ref{prop:hypergeometric-expansion} and Corollary~\ref{lemma:hidden-symmetry-KdS}.

\medskip

\noindent \textit{Step 1: superradiance.} Let $\smlambda{\upalpha}{s}$ be a solution to \eqref{eq:radial-ODE-alpha-KdS-actual-alpha} with outgoing boundary conditions at $\mc{H}^+_c$ and ingoing boundary conditions at $\mc{H}^+$, see Definition~\ref{def:bdry-conditions-KdS}. Let $u(r)=(r^2+a^2)^{1/2}\Delta^{s/2}\smlambda{\upalpha}{s}$, then $u$ solves $u''+Vu=0$, where
\begin{align*}
V&=
\frac{(r^2+a^2)^2\omega^2}{(r^2+a^2)^2}+\frac{2am\omega[\Delta-(r^2+a^2)]}{(r^2+a^2)^2} +\frac{a^2m^2}{(r^2+a^2)^2}-\frac{\Delta (\lambdabar+a^2\Xi^2\omega^2) }{(r^2+a^2)^2\Xi^2}\\
&\qquad -\frac{\Delta}{(r^2+a^2)^4\Xi^2}\lp[a^2\Delta+2Mr(r^2-a^2)\rp]
-\frac{\Delta}{(r^2+a^2)^2\Xi^2}\lp[\lp(\frac{4r^2}{L^2}+\frac14\lp(\frac{d\Delta}{dr}\rp)^2\rp)s^2+\frac{2(\mu-1) r^2}{L^2}\rp]\\
&\qquad-\frac{is}{\Xi(r^2+a^2)} \lp[\lp(\omega-\frac{am}{r^2+a^2}\rp) \frac{d\Delta}{dr} + \frac{4\omega r\Delta}{r^2+a^2}\rp]\,.\numberthis \label{eq:potential-KdS}
\end{align*}

Now consider the following currents associated with the stationary Killing field in Kerr-de Sitter:
\begin{align}
\begin{split}
Q^T[u]&:= \Im \lp(\overline{\omega u }u'\rp)\Rightarrow -(Q^T[u])'=\Im\omega |u'|^2+\Im\lp(\overline{\omega} V\rp)|u|^2\,.
\end{split}\label{eq:QT-superrad-KdS}
\end{align}
Suppose $s=0$ and $\mu=1$. In the potential \eqref{eq:potential-KdS}, both the coefficient of $\lambdabar$  and the $(\omega,m,\lambdabar)$ independent part are non-positive, and the potential is real if $\omega$ is real. Thus, taking into account the boundary conditions on $u$, we deduce the identities:
\begin{gather}\label{eq:superrad-intermediate-KdS}
\begin{split}
\text{if~}\Im\omega>0\,,\quad 0&\geq \Im\omega \int_{-\infty}^\infty \lp( |{u}'|^2+\frac{\lp((r^2+a^2)^2-a^2\Delta\rp)|\omega|^2 -a^2m^2}{(r^2+a^2)^2}|{u}|^2\rp)dr^*\,;\\
\text{if~}\Im\omega=0\,,\quad 0&=\omega(\omega-m\upomega_1)|{u}(-\infty)|^2+ \omega(\omega-m\upomega_2)|{u}(+\infty)|^2\,.
\end{split}
\end{gather}

If $\Im\omega>0$, unless the superradiant condition
\begin{align*}
|\omega|^2<m^2\upomega_1^2 \Leftrightarrow
0<
(\Im\omega)^2< m^2\upomega_1^2-(\Re\omega)^2\,,
\end{align*}
holds, we may infer directly from \eqref{eq:superrad-intermediate-KdS} that $u\equiv 0$, which establishes the result we seek for $s=0$.

If $\omega\in\mathbb{R}$ we cannot argue just from \eqref{eq:superrad-intermediate-KdS}. As in the Kerr case, we need to  appeal to  unique continuation for ODEs such as \eqref{eq:radial-ODE-alpha-KdS-actual-alpha} to deduce that, unless the superradiant condition
\begin{align}
m\neq 0\,, \quad \upomega_2<\frac{\omega}{m}<\upomega_1\,,\label{eq:superradiant-KdS}
\end{align}
holds, we may infer that $u\equiv 0$, thus concluding the proof for $s=0$. In fact, the conclusion can be shown to hold more generally for $s\in\mathbb{Z}_{\leq 2}$ by appealing to the Teukolsky--Starobinsky identities, see \cite{Tachizawa1993}. By the same method, one may deduce that if $|s|\leq 2$ is half-integer, the energy identity implies that in fact $u\equiv 0$ holds independently of $\omega$. However, for if $|s|>2$, for integer and half-integer spins alike, the energy identity may fail to be coercive even for frequencies not in \eqref{eq:superradiant-KdS}, as our previous work \cite{CasalsTdC2021} suggests. 

\medskip
\noindent \textit{Step 2: the $m_2\leftrightarrow m_3$ symmetry.} Let $\smlambda{\tilde R}{s}(z)$  be a solution to \eqref{eq:tilded-ODE-R-KdS} with boundary conditions given by \eqref{eq:case-m2-to-m3-bdry-conditions-KdS}, and define $\tilde{u}:=[z(z-1)(z_2-z)]^{-1/2}\smlambda{\tilde R}{s}(z)$. Exceptionally in this step, we take prime to denote a derivative with respect to $z^*$, a coordinate defined in terms of $z$ through
\begin{align*}
dz^*=\frac{1}{z(z-1)(z_2-z)}dz\implies z^*(z=1)=-\infty\,, \quad z^*(z=z_2)=+\infty\,,
\end{align*}
together with a choice of initial condition. With this notation, $\tilde{u}$ solves $\tilde u''+\tilde V \tilde u=0$, where
\begin{align} \label{eq:KdS-tilded-potential}
\begin{split}
\tilde V&=-\frac{z(z-1)(z_2-z)L^2(\lambdabar+\Xi^2a^2\omega^2-2a\Xi m\omega)}{(r_1-r_0)(r_2-r_3)}- z(z-1)(z_2-z)\lp(z-\frac12\rp) \\
&\qquad-(z-1)(z_2-z)\lp[\frac{s^2(r_1+r_2)^2}{(r_1-r_0)(r_2-r_3)}\lp(z-\frac{(r_2-r_0)(r_1-r_3)}{r_1^2+2r_2 r_1+r_2^2}\rp)+\frac{z(r_0+r_1)^2}{2(r_1-r_0)(r_2-r_3)}\rp]\\
&\qquad-z^2\lp[\eta_2^2(z-1)^2+2(z_2-1)\eta_0\eta_2(z-1) +(z_2-1)^2\eta_0^2\rp]\\
&\qquad-z^2\lp[2(z_2-1)\eta_0\eta_1(z_2-z)+\eta_1^2(z_2-z)^2-2\eta_1\eta_2(z-1)(z_2-z)\rp]\,.
\end{split}
\end{align}
%
%
The microlocal energy current associated with the stationary Killing field in Kerr-de Sitter gives the following identity for $\tilde{u}$:
\begin{align}
\begin{split}
Q^T[\tilde u]&:= \Im \lp(\overline{\omega \tilde u }\tilde u'\rp)\Rightarrow -(Q^T[\tilde u])'=\Im\omega |\tilde u'|^2+\Im\lp(\overline{\omega} \tilde V\rp)|\tilde u|^2\,.
\end{split}\label{eq:QT-mode-stability-KdS}
\end{align}

When $\Im\omega>0$, by the boundary conditions on $\tilde u$, we have 
\begin{align*}
Q^T[\tilde u](-\infty)=Q^T[\tilde u](+\infty)=0\,.
\end{align*}
Furthermore, note that both the $\lambdabar$ coefficient  and the $(\omega,m,\lambdabar)$-independent part of \eqref{eq:KdS-tilded-potential} are non-positive: for the former, this is trivial and for the latter the only non-obvious component can be simplified by Vietà's formula $r_0+r_1+r_2+r_3=0$ through
\begin{align*}
z-\frac{(r_2-r_0)(r_1-r_3)}{r_1^2+2r_2 r_1+r_2^2}=z-1+\frac{(r_1+r_0)^2}{(r_1+r_2)^2}\geq 0\quad \forall\,z\in[1,z_2].
\end{align*}
Thus, we obtain
\begin{align}
\label{eq:ImV-KdS}
\begin{split}
\frac{\Im(\overline{\omega}\tilde V)}{z}& \geq z\Im\lp\{\lp[(\eta_0+\eta_1)(z_2-z)+(\eta_0+\eta_2)(z-1)\rp]^2\overline{\omega}\rp\} \\
&\qquad + (z-1)(z_2-z)\Im\omega \lp[\frac{|\omega|^2-m^2\upomega_1\upomega_2}{\kappa_1\kappa_2}-\frac{L^2\Xi^2a^2}{(r_1-r_0)(r_2-r_3)}|\omega|^2\rp].
\end{split}
\end{align}
It is easy to see that the first term is smallest at $z=z_2$ and that it is non-negative unless we have 
\begin{align} \label{eq:KdS-partial-upper-half}
\frac{|\omega|^2}{m^2}<\lp(\frac{\upomega_2/\kappa_2+\upomega_0/\kappa_0}{1/\kappa_2+1/{\kappa_0}}\rp)^2\,.
\end{align}
This condition in fact ensures that the second term is also non-negative.
Hence, $\Im(\tilde V \overline{\omega})\geq 0$ for $z\in[1,z_2]$ in the full subextremal range of parameters if \eqref{eq:KdS-partial-upper-half} does not hold. In this case, we deduce directly from \eqref{eq:QT-mode-stability-KdS} that $\tilde{u}=0$ and then, using Corollary~\ref{lemma:hidden-symmetry-KdS}(i) if $s\leq 0$ or Corollary~\ref{lemma:hidden-symmetry-KdS}(ii) if $s\geq 0$, we conclude that $\smlambda{R}{s}\equiv 0$.

Now consider $\Im\omega=0$. The potential \eqref{eq:KdS-tilded-potential} is real and so, from \eqref{eq:QT-mode-stability-KdS} and the boundary conditions on $\tilde u$, we obtain
\begin{align*}
&\omega\lp[(\omega-m\upomega_1)-(\omega-m\upomega_0)\frac{\kappa_1}{\kappa_0}\rp]|\tilde{u}(-\infty)|^2+\omega\lp[(\omega-m\upomega_2)+(\omega-m\upomega_0)\frac{\kappa_2}{\kappa_0}\rp]|\tilde{u}(+\infty)|^2=0\,.
\end{align*}
(Note that the factors of $\omega$ outside the square brackets can be replaced by any linear combination of $\omega$ and $m$, and so we need not impose $\omega\neq 0$ in what follows.) Thus, unless we have
\begin{align}\label{eq:KdS-partial-real}
m\neq 0\,, \quad \frac{\upomega_1/\kappa_1-\upomega_0/\kappa_0}{1/{\kappa_1}-1/{\kappa_0}}<\frac{\omega}{m}<\frac{\upomega_2/\kappa_2+\upomega_0/\kappa_0}{1/\kappa_2+1/{\kappa_0}}\,, 
\end{align}
we may infer, as before, by a unique continuation argument, that $\tilde{u}\equiv 0$.  As in the case $\Im\omega>0$, it then follows from  Corollary~\ref{lemma:hidden-symmetry-KdS} that $\smlambda{R}{s}\equiv 0$.

Finally, we use the Vietà formula $\sum_{i\neq j}r_ir_j=a^2-L^2$ to rewrite the factors appearing in conditions \eqref{eq:KdS-partial-upper-half} and \eqref{eq:KdS-partial-real}:
\begin{align*}
\frac{\upomega_1/\kappa_1-\upomega_0/\kappa_0}{1/{\kappa_1}-1/{\kappa_0}} &= \frac{2a}{L^2\Xi-(r_0+r_1)^2}=\upomega_2+\frac{2a\kappa_2}{(r_2-r_3)}\lp(1-\frac{(r_0+r_1)^2}{L^2\Xi}\rp)^{-1}> \upomega_2\,,\\
\frac{\upomega_2/\kappa_2+\upomega_0/\kappa_0}{1/\kappa_2+1/{\kappa_0}} &=
\frac{2a}{L^2\Xi-(r_0+r_2)^2}=\upomega_1-\frac{2a\kappa_1}{(r_1-r_3)}\lp(1-\frac{(r_0+r_2)^2}{L^2\Xi}\rp)^{-1}< \upomega_1\,,
\end{align*}
as long as $a\neq 0$, since $(r_0+r_1)^2,(r_0+r_2)^2<L^2\Xi$.
\end{proof}


\subsection{Epilogue: mass symmetries within the MST method}

In this section, we provide an alternative proof of Corollary~\ref{lemma:hidden-symmetry-KdS} based not on the global hypergeometric expansions of Proposition~\ref{prop:hypergeometric-expansion} but on the matching of local hypergeometric expansions. In the study of quasinormal modes on Kerr-de Sitter, this method was first introduced by Suzuki, Tagasuki and Umetsu \cite{Suzuki1999,Suzuki2000},  but it is known as the MST method after the work \cite{Mano1996} in Kerr. 

Let us fix $M>0$, $|a|<M$, $s\in\frac12\mathbb{Z}$, $m-s\in\mathbb{Z}$, $\omega\in\mathbb{C}\backslash\{0\}$ with $\Re\omega\geq 0$ and $\Im\omega\geq 0$, and $\lambdabar\in\mathbb{C}$, assuming additionally that $s\leq 0$ if $\Re\omega=m\upomega_+$. To aid the reader, we write the MST quantities $\gamma$, $\delta$, $\epsilon$, $\omega_H$ and $v$ in \cite{Suzuki1999,Suzuki2000} in terms of the quantities identified in \eqref{eq:SQCD-masses-KdS}:
\begin{gather*}
\delta := 1\pm (m_1-m_2)\,, \quad \gamma := 1-m_1-m_2\,, \quad \epsilon:=1-m_3-m_4\,, \qquad \omega_H:=1 - m_1 (1 \mp 1) - m_2 (1 \pm 1)\,,\\
\sigma_-:=  1-m_1\frac{1\mp 1}{2}-m_2\frac{1\pm 1}{2}-m_3\,, \quad \sigma_+:=1-m_1\frac{1\mp 1}{2}-m_2\frac{1\pm 1}{2}-m_4\,,\numberthis \label{eq:Heun-parameters-MST}\\
v:=-1 + z_2 - 4 E z_2-\lp[(\gamma - \delta) \epsilon + 
   2 \sigma_+ \sigma_- + \delta (1 - \omega_H) + (\omega_H^2-1)(1-z_2)\rp]\,, \quad z_r:=1-x_r=z_2\,.
\end{gather*}
Note that, in the notation of \cite{Suzuki1999}, we have taken the minus sign for the exponent at $x=0$ and the plus sign for the exponent at $x=x_r$, i.e.\ $A_3=A_{3+}$; the latter choice is different from the choice made in \cite{Suzuki2000} but it is consistent with the boundary conditions we consider. We also allow the sign of the exponent at $x=-1$ to be either plus or minus, which corresponds to the two choices in \eqref{eq:Heun-parameters-MST}. Within the MST formalism, existence of a non-trivial solution to \eqref{eq:radial-ODE-alpha-KdS} with outgoing boundary conditions at $\mc{I}^+$ and ingoing boundary conditions at $\mc{H}^+$ is, by \cite[Equation 3.8a]{Suzuki2000} and \cite[Equation 3.29]{Suzuki1999}, equivalent to the condition
\begin{align}
\frac{B}{A}\lp(\frac{C_\nu}{D_\nu}+\frac{C_{-\nu-1}}{D_{-\nu-1}}\rp)=0\,,\label{eq:MST-qnm-condition-KdS}
\end{align}
where we have used the shorthand notation
\begin{gather*}
A:=\prod_{j=1}^4\Gamma(1+m_j)\,, \qquad
B:=z_2^{-m_1(1\mp 1)/4-m_2(1\pm 1)/4}\Gamma(1-m_1-m_2)\Gamma\lp(1 + m_3 - m_1 \frac{1 \mp 1}{2} - m_2 \frac{1 \pm 1}{2}\rp),
\\
C_\nu:=\sum_{n=0}^{\infty}\frac{(-1)^n\Gamma(N+\nu+1)}{n!}\frac{b_n^{\nu}}{\prod_{j=1}^4\Gamma(1+N-m_j)}\,, \qquad
D_\nu:=\sum_{n=-\infty}^0\frac{b_n^{\nu}}{\prod_{j=1}^4\Gamma(1+N+m_j)}\,,
\end{gather*}
writing $N:=n+\nu$ and where the coefficients $b_n^\nu$, obtained from the $a_n^\nu$ coefficients of \cite[Equation 3.1]{Suzuki1999} through $b_n^{\nu}=\Gamma\left(1+m_1+N\right)\Gamma\left(1+m_2+N\right)\Gamma\left(1-m_3+N\right)\Gamma\left(1-m_4+N\right)a_n^{\nu}$, satisfy the following recursive relation, c.f.\ \cite[Equation 3.2]{Suzuki1999}:
\begin{gather*}
\alpha_n^\nu b_{n+1}^{\nu}+\beta_n^{\nu}b_n^{\nu}+\gamma_n^\mu b_{n-1}^{\nu}
=0\,,\\
\alpha_n^\nu :=\frac{1}{2(1+N)(3+2N)}\,, \qquad
\gamma_n^\nu :=\frac{\left(\prod_{j=1}^4(N^2-m_j^2)\right)}{2N(-1+2N)}\,, \qquad \beta_n^\nu :=  \lp(E-\frac14\rp) z_2+\frac14\,.
\end{gather*}
Furthermore, the parameter $\nu$ is the so-called renormalized angular momentum parameter and its value is chosen to ensure that the continued fraction equation
\begin{align}
\frac{\alpha_{n-1}^\nu\gamma_{n}^\nu}{\lp(\beta_n^\nu -\frac{\alpha_n^\nu\gamma_{n+1}^\nu}{\beta_{n+1}^\nu-\cdots}\rp)\lp(\beta_{n-1}^\nu -\frac{\alpha_{n-2}^\nu\gamma_{n-1}^\nu}{\beta_{n-2}^\nu-\cdots}\rp)}=1\label{eq:MST-nu-equation-KdS}
\end{align}
holds for an arbitrary choice of $n\in\mathbb{Z}$, c.f.\ \cite[Equation 3.6]{Suzuki1999}.  Noting that $B\neq 0$, this choice ensures that \eqref{eq:MST-qnm-condition-KdS} is equivalent to 
\begin{align}
\frac{1}{A}\lp(\frac{C_\nu}{D_\nu}+\frac{C_{-\nu-1}}{D_{-\nu-1}}\rp)=0\,. \label{eq:MST-qnm-condition-simplified-KdS}
\end{align}

Clearly, $\alpha_{n}^\nu$, $\beta_{n}^\nu$ and $\gamma_n^\nu$ are all separately invariant under the map $(m_1,m_2,m_3,m_4)\mapsto(m_i,m_j,m_k,m_l)$ with $i\neq j\neq k\neq l$, and so $b_n^\nu$ and $\nu$ must also be preserved. Consequently, the same is true for $C_\nu$ and $D_\nu$; furthermore, $A$ is clearly also preserved. Hence, we conclude  that the condition \eqref{eq:MST-qnm-condition-simplified} is invariant under the map $(m_1,m_2,m_3,m_4)\mapsto(m_i,m_j,m_k,m_l)$ with $i\neq j\neq k\neq l$. By using the fact that \eqref{eq:radial-ODE-alpha-KdS} and the boundary conditions are invariant under taking at once $\Re\omega\mapsto-\Re\omega, m\mapsto -m$ and complex conjugation, we arrive at the same conclusion for $\Re\omega\leq 0$. The statement of Corollary~\ref{lemma:hidden-symmetry-KdS} then follows easily from exploiting these symmetries.


\bibliographystyle{halpha-abbrv-rita}
{\small \bibliography{unpub.bib,../../library.bib}}

\end{document}